\documentclass[letterpaper]{article} 
\usepackage[preprint]{aaai2027}  
\usepackage[hyphens]{url}  
\usepackage{graphicx} 
\usepackage{natbib}  
\usepackage{caption} 
\usepackage{amsmath}
\usepackage{amssymb}
\usepackage{mathtools}
\usepackage{amsthm}
\usepackage{etoolbox}
\usepackage{xparse}
\usepackage{booktabs} 
\usepackage{tabularx} 
\usepackage{enumitem}
\usepackage{microtype}

\usepackage{float}
\floatstyle{ruled}
\newfloat{algorithm}{t}{loa}
\floatname{algorithm}{Algorithm}
\newlist{algsteps}{enumerate}{1}
\setlist[algsteps]{label=\footnotesize\arabic*:,wide=0pt,labelsep=0.4em,
  itemsep=1pt,topsep=1pt,parsep=0pt}

\usepackage[disable,textsize=tiny]{todonotes}

\newenvironment{tightalign*}{%
  \setlength{\abovedisplayskip}{0pt}%
  \setlength{\belowdisplayskip}{0pt}%
  \setlength{\abovedisplayshortskip}{0pt}%
  \setlength{\belowdisplayshortskip}{0pt}%
  \csname align*\endcsname
}{\csname endalign*\endcsname}

\theoremstyle{plain}
\newtheorem{theorem}{Theorem}[section]
\newtheorem{proposition}[theorem]{Proposition}

\theoremstyle{definition}
\newtheorem{definition}[theorem]{Definition}

\theoremstyle{remark}

\title{Identifiability-Aware Source Apportionment in City-Scale Advection-Diffusion Systems}

\author{
    Ankit Bhardwaj,
    Lakshmi Subramanian
}
\affiliations{
    Department of Computer Science, New York University\\
    bhardwaj.ankit@nyu.edu
}

\begin{document}

\maketitle

\begin{abstract}
Source apportionment from sparse urban air-quality sensors is an inverse problem confounded by sensor placement, wind-driven transport, background variation, and noise.  Emission inventories restrict the source field to a few candidate groups, but a good predictive fit does not mean those groups can be told apart from the data.  Standard receptor models report confident attributions even when they cannot.
We cast inventory-based apportionment as a wind-conditioned lagged inverse problem in which each source and temporal-basis coefficient leaves a distinct sensor-time fingerprint.  We make identifiability a reported output.  After projecting out a low-dimensional background, separability is governed by the projected response matrix \(\widetilde H_\Phi\).  Exact recovery needs full column rank, and noise-robust attribution a margin in its singular values, coefficient visibility, background absorption, and pairwise coherence.  Our identifiability-aware apportionment (IASA) framework fits nonnegative coefficients, reconstructs activity trajectories, quantifies uncertainty, and, crucially, diagnoses indistinguishable sources and falls back to conservative groupings.  Across controlled degeneracies (wind-geometry collapse, source-like backgrounds) IASA flags every non-identifiable case and recovers apportionment to within \(0.01\) share error.  NNLS, CMB, and PMF baselines detect none and misattribute by up to \(1.2\).  On a New Delhi platform (government PM\(_{2.5}\) and wind records, regulatory sensor locations, four proxy source groups), IASA reports the finest attribution the declared inventories, transport, background, lag, and noise can support, and no more.
\end{abstract}

\section{Introduction}
\label{sec:introduction}

Urban air pollution is driven by local emissions, meteorological transport,
regional background, and sensor noise.  For policy the central question is often
not only \emph{where} pollution is high but \emph{which} source groups (traffic,
industry, brick kilns, population-related activity) produced it, given source
inventories, transport models, and sparse measurements.

This is especially hard in sparse sensor networks: a model may fit sensor
trajectories accurately yet fail to identify the underlying contributions.  The
reason is geometric.  After wind transport and sparse observation, two distinct
source groups can induce nearly identical sensor-time signatures.  A strong
distant source and a weaker nearby one become observationally equivalent after
attenuation, travel time, background correction, and noise.
Fine-grained attribution is then not merely uncertain; it is not identifiable
from the sensing system.

We study this in the inventory-based setting.  Rather than recovering an arbitrary
source field, we assume a finite set of known or proxy source maps that restrict
attribution to a low-dimensional, interpretable model.  Each source's
time variation is represented as a nonnegative combination of a few temporal basis
functions.  The unknowns are then source--basis coefficients, from which activity
trajectories and contributions are reconstructed.  These restrictions alone do not
make the problem identifiable: coefficient fingerprints can still be weakly
visible, absorbed by background components, or mutually coherent after transport.

The key object is a projected lagged source-response matrix.  For a fixed wind
realization, transport operator, inventory, sensor layout, and lag window, each
source--basis pair induces a finite-horizon sensor-space fingerprint; stacking
them gives \(H_{\Phi}^{\mathrm{lag}}\).  After projecting out a low-dimensional
background basis \(Q\), the inverse problem is governed by
\(\widetilde H_\Phi = P_Q^\perp H_{\Phi}^{\mathrm{lag}}\).  Its rank, singular
values, coefficient visibility, background absorption, pairwise coherence, and ray
distance determine which contributions can be separated and which should be
reported only after merging or qualification.  The goal is therefore not the most
detailed contribution vector the inventory permits.  It is a conservative
resolution, one defensible under the available sensors, wind, background, and
noise, and conditional on the declared source model.  Estimates are reported
alongside identifiability diagnostics and ambiguity certificates.

We make the following contributions:
\begin{itemize}[leftmargin=*,itemsep=0pt,topsep=0pt]
    \item We formulate sparse-sensor, inventory-based source apportionment as a
    lagged wind-conditioned inverse problem over nonnegative
    source--temporal-basis coefficients, and identify the projected lagged
    response matrix \(\widetilde H_\Phi=P_Q^\perp H_{\Phi}^{\mathrm{lag}}\) as the
    central object governing coefficient separability after transport, sparse
    observation, lag, and background correction.
    \item We turn exact and noise-robust identifiability into a concrete
    reporting protocol built from computable diagnostics of \(\widetilde H_\Phi\)
    (rank, smallest singular value, condition number, effective rank, coefficient
    visibility, background absorption, pairwise fingerprint coherence, and
    \mbox{ray distance}), keeping what is fitted separate from what is diagnosed
    so that predictive accuracy cannot be mistaken for identifiability.
    \item We propose an identifiability-aware apportionment (IASA) procedure that
    fits nonnegative source--basis coefficients, reconstructs source activity
    trajectories, reports uncertainty, flags weakly visible coefficients,
    recommends connected report groups when eligible fingerprints are
    indistinguishable, and adds a per-sensor spatial attribution that back-traces
    the transport response to upwind regions.
    \item We instantiate the framework on a New Delhi platform (government
    PM\(_{2.5}\) and wind records, regulatory sensor locations, kernel
    coordinate-query gridded wind imputation, and four normalized proxy source
    groups with declared temporal activity bases) and define controlled and
    observed-data evaluations of conditioning, source ambiguity, background
    stress, wind diversity, temporal-basis recovery, response error, uncertainty,
    and residual structure. Across controlled degeneracies, IASA flags every non-identifiable case and recovers apportionment to within \(0.01\) share error, outperforming the baselines: NNLS, CMB, and PMF.
\end{itemize}

\section{Related Work}
\label{sec:related_work}

Source apportionment has a long history in air-quality science, spanning
receptor-oriented methods (chemical mass balance and factor-analytic approaches
such as positive matrix factorization
\citep{watson2002receptor,paatero1994positive,hopke2016review}) and
source-oriented models that propagate emissions inventories through transport,
with reviews and harmonization efforts documenting their practical sensitivity to
source profiles, tracer collinearity, factor interpretation, and inventory
assumptions
\citep{viana2008source,belis2013critical,belis2019european,mircea2020european,hopke2020global}.
This ambiguity is acute in settings such as India, where the same city yields
divergent source-contribution estimates and heterogeneously reported categories
\citep{pant2012critical,karagulian2015contributions}.
Estimating latent quantities from partial observations belongs to inverse
problems and data assimilation
\citep{tarantola2005inverse,engl1996regularization,wang2023kalman,carrassi2018data}
and, more recently, physics-informed and operator-learning methods
\citep{raissi2019physics,li2024physics,li2020fourier,bhardwaj2025fieldformer}.
These motivate our transport-structured constrained backend but do not themselves
resolve identifiability in a many-to-one sparse-sensing map.
Graph, transformer, and sequence models address the same sparse-sensor regime
\citep{li2017diffusion,yu2017spatio,wu2019graph,chen2023group,iyer2022modeling,bhardwaj2025comprehensive},
but target forecasting, interpolation, or representation learning rather than
attribution to physically meaningful inventories.  Finally, our separability
question is a form of the observability studied in control theory and system
identification \citep{chen1984linear,ljung1987theory,banks2012estimation,colton1990inverse}.
Unlike this prior work, we analyze when declared or proxy source groups are
\emph{distinguishable} from sparse concentration sensors after wind-conditioned
transport, finite lags, background correction, and noise, and report the
attribution resolution the sensing system supports.  See Appendix~\ref{app:extended_related_work}
for details.

\section{Problem Setup}
\label{sec:problem_setup}

We formalize source apportionment as a finite-dimensional inverse problem in which
known source inventories are transported to a sparse sensor network under a
wind-conditioned lagged operator.  This defines the observation model, inventory
parameterization, response construction, background correction, and the projected
response matrix \(\widetilde H_\Phi\) used throughout.  Transport, background, and fit
details are deferred to Appendix~\ref{app:transport_response}.

\noindent\textbf{City-scale transport model.}
Let \(u(\mathbf x,t)\) be the pollutant concentration field on
\(\Omega\subset \mathbb R^d\), evolving as \(\partial u/\partial t
= \mathcal F_\eta(u,\mathbf x,t;w(t)) + s(\mathbf x,t)+b(\mathbf x,t)\) under
meteorological state \(w(t)\), local emission \(s\), background \(b\), and
transport/dispersion parameters \(\eta\); the dynamics may be realized by a plume
model, an advection--diffusion solver, or a transport-constrained surrogate.

\noindent\textbf{Sparse observations.}
Discretize the concentration field on \(n\) grid cells, so that
\(\mathbf u_t\in \mathbb R^n\).  Let
\(X_{\mathrm{sens}}=\{\mathbf x_1,\ldots,\mathbf x_m\}\) be the fixed sensor
locations with \(m\ll n\).  The observation operator
\(O:\mathbb R^n\to \mathbb R^m\) samples or interpolates the latent field at
these locations, \(\mathbf y_t = O\mathbf u_t + \boldsymbol\epsilon_t\); stacking
over \(t=1,\ldots,T\) gives
\(Y = [\mathbf y_1^\top,\ldots,\mathbf y_T^\top]^\top \in \mathbb R^{mT}\).
Because \(m\) is small relative to the field dimension, many latent fields and
source configurations fit the same sensor~trajectory.

\noindent\textbf{Inventory-based source parameterization.}
We assume a set of \(K\) known or proxy source maps,
\(S = [\mathbf s_1\;\cdots\;\mathbf s_K]\in \mathbb R^{q\times K}\),
where \(\mathbf s_k\in\mathbb R^q\) is the spatial pattern of source group
\(k\).  Time-varying source activity uses a fixed, predeclared nonnegative
temporal dictionary
\(\Phi=[\boldsymbol\phi_1\;\cdots\;\boldsymbol\phi_B]\in\mathbb R_+^{T\times B}\)
and coefficients \(C=[c_{kb}]\in\mathbb R_+^{K\times B}\), so the activity of
source \(k\) is \(\theta_k(t)=\sum_{b=1}^{B}c_{kb}\phi_b(t)\ge0\).  The basis is a
shared dictionary of admissible temporal patterns, not a claim that all sources
follow one trajectory: each source has its own nonnegative weights, and
source-specific knowledge is imposed by masking inadmissible coefficients.  We do
not learn \(\Phi\) from the same sparse observations, since jointly estimating
\(\Phi\) and \(C\) would add bilinear non-identifiability beyond the source
distinguishability analyzed here.  We vectorize \(C\) in source-major, basis-minor order as
\(\mathbf c=\operatorname{vec}_{\mathrm{src}}(C)\in\mathbb R_+^J\), \(J=KB\); the
constant-activity model is the special case \(B=1\), \(\phi_1(t)=1\).

\noindent\textbf{Wind-conditioned lagged response.}\label{subsec:wind_field_estimation}
Government wind direction records where wind comes from, whereas transport needs
the direction pollution moves; we convert direction \(\mathit{WD}\) and speed
\(\mathit{WS}\) to transport vectors
\begin{equation}
(U_x,V_y)=-\mathit{WS}\,(\sin\alpha,\cos\alpha),\;\;
\alpha=\mathit{WD}\,\pi/180,
\label{eq:wind_direction_conversion}
\end{equation}
and complete them on the response grid with a normalized Gaussian-kernel
coordinate-query imputer to obtain a gridded transport field
\(\widehat{\mathbf w}_t(\mathbf x_g)\); and the same interface
supplies synthetic controlled winds and transport ensembles kept separate from
inventory scenarios.  Emissions released at \(t-\ell\) affect sensors at \(t\)
only after transport by the intervening wind.  Let
\(G_{t,t-\ell}^{\partial}(w_{t-\ell:t})\in \mathbb R^{n\times q}\) map emissions
released on the source grid at \(t-\ell\) to the concentration grid at \(t\); the
superscript \(\partial\) denotes an open-boundary operator (mass leaving
\(\Omega\) is removed, not reflected or wrapped).\label{subsec:plume_response} It
is realized by a Gaussian puff approximation evaluated at sensor coordinates that
preserves nonnegativity and is mass non-increasing inside the domain, with
dispersion, exit bookkeeping, and the pre-fit initial-condition baseline given in
Appendix~\ref{app:transport_response}.  For a maximum lag \(L\),
\(\mathcal L_t=\{0,\ldots,\min(L,t-1)\}\), the unit-coefficient sensor fingerprint
of source \(k\) and basis component \(b\) at time \(t\) is
\begin{equation}
\mathbf h_{t,kb}^{\mathrm{lag}}
= \sum_{\ell\in\mathcal L_t}
\phi_b(t-\ell)\,O G_{t,t-\ell}^{\partial}(w_{t-\ell:t})\mathbf s_k
\in \mathbb R^m.
\label{eq:constructed_response}
\end{equation}
Stacking over time in source-major, basis-minor order gives
\(H_{\Phi}^{\mathrm{lag}}\in \mathbb R^{mT\times J}\), \(J=KB\), each column the
finite-horizon sensor-time pattern produced by one unit of a coefficient.  The lag
window adds no sensor--time rows and is chosen before fitting from physical travel-
and residence-time bounds: for adjacent candidates \(L\), \(L+\Delta\) with the
same rows and columns,
\begin{equation}
\eta_L=
\frac{\|H_{\Phi}^{\mathrm{lag}}(L+\Delta)
-H_{\Phi}^{\mathrm{lag}}(L)\|_F}
{\max\{\|H_{\Phi}^{\mathrm{lag}}(L+\Delta)\|_F,\epsilon\}},
\label{eq:lag_convergence}
\end{equation}
and the primary lag is the smallest candidate with \(\eta_L\le\tau_L\) (default
\(10^{-3}\)); the fitted \(\mathbf c\) never selects \(L\).

\noindent\textbf{Background correction and projection.}\label{subsec:background_construction}
Observed concentrations include components not explained by the local inventory
(regional background, smooth trends, sensor offsets), giving
\(Y = H_{\Phi}^{\mathrm{lag}}\mathbf c + Q\boldsymbol\gamma + E\) with a
low-dimensional basis \(Q\in \mathbb R^{mT\times r}\) specified before fitting from
source-independent metadata only (timestamps, day labels, sensor identity, and
coordinates; effective rank capped at eight, the rank-four for New Delhi observations),
so apportionment is not confounded by background correction; columns built from
inventories, fingerprints, or fitted signals are excluded and permitted only in
labeled stress tests (Appendix~\ref{app:transport_response}).  Real records need
not be complete: let \(\mathcal O\) be the ordered observed sensor--time
PM\(_{2.5}\) rows and \(M_{\mathcal O}\in\{0,1\}^{N\times mT}\) select them, applied
identically to every row-aligned object,
\begin{equation}
Y_{\mathcal O}=M_{\mathcal O}Y,\quad
H_{\Phi,\mathcal O}^{\mathrm{lag}}
=M_{\mathcal O}H_{\Phi}^{\mathrm{lag}},\quad
Q_{\mathcal O}=M_{\mathcal O}Q,
\label{eq:observation_selection}
\end{equation}
with an alignment mismatch treated as an error;
wind may be imputed, PM\(_{2.5}\) is not, and \(M_{\mathcal O}=I\) for complete
controlled data.  Dropping the \(\mathcal O\) subscript, let
\(P_Q^\perp = I - QQ^\dagger\) project onto the orthogonal complement of the
background space (applied implicitly through a thin SVD of \(Q\)); the corrected
inverse problem is
\begin{equation}
\widetilde Y = \widetilde H_\Phi\mathbf c + \widetilde E,
\qquad
\widetilde Y=P_Q^\perp Y,
\qquad
\widetilde H_\Phi=P_Q^\perp H_{\Phi}^{\mathrm{lag}},
\label{eq:projected_model}
\end{equation}
so each projected fingerprint
\(\widetilde{\mathbf h}_{kb}=P_Q^\perp\mathbf h_{kb}^{\mathrm{lag}}\) encodes the
observable effect of one source--basis coefficient after transport, lag, sparse
sensing, and background correction.  All identifiability diagnostics in this paper
are computed from \(\widetilde H_\Phi\), not the raw inventory or the
\mbox{unprojected response}.

\noindent\textbf{Source activity fit.}
Scientific knowledge may declare a set \(\mathcal F_0\) of coefficients identically
zero before the response is diagnosed or fitted; with
\(\mathcal I=\{1,\ldots,J\}\setminus\mathcal F_0\), the fit is computed on
\(\widetilde H_{\Phi,\mathcal I}\) and zeros are restored afterward (default
\(\mathcal F_0=\varnothing\)).  Coefficients are never deleted after fitting merely
for being small, which would change the diagnosed problem after seeing the data.
Source--basis coefficients are estimated by nonnegative regularized least squares
\begin{equation}
\widehat{\mathbf c}
=\arg\min_{\mathbf c\ge 0}
\|\widetilde Y-\widetilde H_\Phi\mathbf c\|_2^2
+\lambda R(\mathbf c),
\label{eq:nnls_estimator}
\end{equation}
with \(R\) a ridge or weak inventory-prior penalty, reshaped into activity
trajectories \(\widehat\theta_k(t)=\sum_{b}\widehat c_{kb}\phi_b(t)\).  The
regularizer choices, an equivalent joint source/background fit, and an optional
constrained end-to-end refinement of wind, dispersion, and coefficients, accepted
only if it does not degrade identifiability, are given in
Appendix~\ref{app:transport_response}.

\section{Identifiability-Aware Source Apportionment}
\label{sec:identifiability_aware_method}
\label{sec:identifiability_theory}

We now characterize when source activity at a chosen temporal-basis resolution
is determined by the projected observations.  Throughout,
\(\widetilde Y=\widetilde H_\Phi\mathbf c+\widetilde E\) with
\(\widetilde H_\Phi=P_Q^\perp H_{\Phi}^{\mathrm{lag}}\in\mathbb R^{N\times J}\),
\(J=KB\), conditional on the constructed wind field, transport operator,
inventories, temporal basis, lag window, and background basis.  We use the
flattened coefficient index \(j=(k,b)\).

\begin{definition}[Exact identifiability]
The source--basis coefficient vector is identifiable from \(\widetilde Y\) if
\(\widetilde H_\Phi\mathbf c_1=\widetilde H_\Phi\mathbf c_2\Longrightarrow
\mathbf c_1=\mathbf c_2\); this is identifiability of the reconstructed source
activity trajectories at the selected temporal-basis resolution.
\end{definition}

\begin{proposition}[Identifiability of inventory-based apportionment]
\label{prop:rank_identifiability}
In the noiseless projected model \(\widetilde Y=\widetilde H_\Phi\mathbf c\),
\(\mathbf c\) is identifiable uniformly for all
\(\mathbf c\in\mathbb R_+^J\) if and only if
\begin{equation}
\operatorname{rank}(\widetilde H_\Phi)=J.
\label{eq:rank_condition}
\end{equation}
\end{proposition}

The proof (rank--nullity in both directions, with a nonnegative-feasible
counterexample when the rank is deficient) is in
Appendix~\ref{app:identifiability_proofs}.  We use this full-column-rank
condition as our identifiability criterion for the declared coefficient set.  Coefficients fixed to zero by scientific knowledge are masked
before fitting; the theorem and all diagnostics then apply to the reduced,
predeclared column set, and a near-zero fitted estimate never justifies deleting
its column after fitting.  Rank deficiency may arise from spatial source
ambiguity, dependence among temporal modes, or their interaction.

\noindent\textbf{Noise-Robust Identifiability.}
Let \(\sigma_1(\widetilde H_\Phi)\ge\cdots\ge
\sigma_J(\widetilde H_\Phi)\ge0\) be the singular values, padded with zeros when
\(J>N\) or the matrix is rank deficient.  Even at full rank, small singular
values amplify noise.

\begin{proposition}[Noise-robust apportionment]
\label{prop:noise_robust}
Assume \(\widetilde Y=\widetilde H_\Phi\mathbf c+\widetilde E\) and
\(\operatorname{rank}(\widetilde H_\Phi)=J\).  Let
\(\widehat{\mathbf c}_{\mathrm{LS}}\) be the unconstrained least-squares
estimate.  Then
\(\|\widehat{\mathbf c}_{\mathrm{LS}}-\mathbf c\|_2
\le\|\widetilde E\|_2/\sigma_J(\widetilde H_\Phi)\).
If \(\widehat{\mathbf c}\) is the nonnegative least-squares estimate and
\(\mathbf c\ge0\), then
\begin{equation}
\|\widehat{\mathbf c}-\mathbf c\|_2
\le
\frac{2\|\widetilde E\|_2}{\sigma_J(\widetilde H_\Phi)}.
\label{eq:nnls_noise_bound}
\end{equation}
\end{proposition}

Both bounds follow from the pseudoinverse identity
\(\|\widetilde H_\Phi^\dagger\|_2=1/\sigma_J(\widetilde H_\Phi)\) together with
feasibility of the true \(\mathbf c\) in the nonnegative program
(Appendix~\ref{app:identifiability_proofs}).  Thus
\(\sigma_J(\widetilde H_\Phi)\) measures worst-case robustness: apportionment can be
exactly identifiable yet practically unstable when it is small.

\noindent\textbf{Identifiability Diagnostics.}
Beyond exact rank, IASA reports a fixed panel of diagnostics computed from
\(\widetilde H_\Phi\), each with a predeclared, noise- or physics-derived
threshold that is never tuned on source recovery (full definitions in
Appendix~\ref{app:diagnostics_definitions}, threshold values in
Table~\ref{tab:thresholds}, and a summary in Table~\ref{tab:diagnostics}).  The
\emph{padded minimum singular value} \(\sigma_J\) is the primary score: a
rank-deficient response is padded to zero and never reported as stable, and the
condition number is \(\kappa=\sigma_1/\sigma_J\).  The \emph{effective rank}
\(r_{\mathrm{eff}}(\tau_\sigma)\) counts directions resolvable above the
observation-noise level \(\tau_\sigma\); when it is below \(J\) the network
cannot support all \(J\) coefficients at that noise level.  \emph{Coefficient
visibility} \(v_j=\|\widetilde{\mathbf h}_j\|_2\) is the post-projection
fingerprint magnitude, and the weak set \(\mathcal W=\{j:v_j\le\tau_v\}\) flags
coefficients below the minimum detectable signal.  For eligible (nonweak) pairs,
\emph{pairwise coherence}
\(\rho_{ij}=|\widetilde{\mathbf h}_i^\top\widetilde{\mathbf h}_j|/
(\|\widetilde{\mathbf h}_i\|_2\|\widetilde{\mathbf h}_j\|_2)\) measures
near-proportional signatures, giving the ambiguous-pair set
\(\mathcal A=\{(i,j):i,j\notin\mathcal W,\ \rho_{ij}>\tau_\rho\}\); an equivalent
scale-invariant ray distance \(\sqrt{1-\rho_{ij}^2}\) encodes the same geometry.
\emph{Background absorption}
\(a_j=\|P_Q\mathbf h_j^{\mathrm{lag}}\|_2/\|\mathbf h_j^{\mathrm{lag}}\|_2\) is the
source-signal fraction lying in the background space, which is why
identifiability is assessed after projection.  Finally, response error, whether
from transport or inventory, acts like additional observation error: it inflates
the \(\sigma_J\) bound and is most damaging when \(\widetilde H_\Phi\) is already
ill-conditioned, with transport error propagated into uncertainty intervals and
inventory error reported only as named robustness scenarios
(Appendix~\ref{app:operator_error_details}).

\noindent\textbf{Identifiable Source Resolution.}
When coefficient fingerprints are weak or highly coherent, separate source
estimates can mislead.  We define the source ambiguity graph
\(\mathcal E_{\mathrm{src}}
=\{(k,k'):\exists b,b'\ ((k,b),(k',b'))\in\mathcal A\}\); each edge retains the
maximum eligible coefficient coherence, minimum ray distance, and the
source--basis pair attaining the trigger.  Its deterministically ordered
connected components \(\mathcal G=\{G_1,\ldots,G_M\}\) are the recommended report
groups.  Weak coefficients remain separate flags and create no edges; rank
deficiency with no eligible edge produces a global unresolved warning rather than
an invented merge.  These components are conservative reporting units, not the
finest identifiable partition: they can transitively over-merge (edges \(A\)--\(B\)
and \(B\)--\(C\) group all three even when \(A\) and \(C\) are distinguishable),
so the edge table retains both triggering pairs and their diagnostic values to
keep this visible.

The coefficient fit is not silently replaced by a grouped refit.  For reporting,
group activity and fitted sensor contribution are the member sums
\(\widehat\theta_{G_a}(t)=\sum_{k\in G_a}\widehat\theta_k(t)\) and
\(\widehat Y_{G_a}=\sum_{k\in G_a}\sum_b
\mathbf h_{kb}^{\mathrm{lag}}\widehat c_{kb}\), while matrix diagnostics remain
coefficient level.  Table~\ref{tab:diagnostics} in
Appendix~\ref{app:temporal_basis_notation} summarizes the scalar diagnostics and
the set- and graph-valued reporting outputs.

\noindent\textbf{Per-sensor source footprints.}
Policy questions are often also spatial.  Because the response is linear in
\(\mathbf c\), the fitted solution already localizes each monitor's signal
without a new model: the per-sensor footprint is the puff response of
Section~\ref{subsec:plume_response} read backward from the sensor, and
identifiability is inherited rather than created
(\(\sigma_J(\widetilde H_\Phi^{(s)})\le\sigma_J(\widetilde H_\Phi)\), so a single
sensor is never more identifiable than the pooled network).  Footprints are
therefore reported as spatial overlays of the global fit, with per-sensor source
\emph{shares} asserted only at the globally identifiable resolution and
aggregated to a report group wherever the pooled diagnostics require one
(Appendix~\ref{app:footprint_details}).

\noindent\textbf{IASA Algorithm.}
We combine the estimator of
Section~\ref{sec:problem_setup} with the diagnostics developed above into an identifiability-aware source
apportionment procedure (IASA) that estimates source contributions only at the
resolution supported by the projected response matrix.  It takes the sparse
observations and their row mask, sensors \(X_{\mathrm{sens}}\), wind observations
and masks, inventories \(S\in\mathbb R^{q\times K}\), background basis \(Q\),
maximum lag \(L\), temporal basis \(\Phi\in\mathbb R_+^{T\times B}\), and transport
response \(G^\partial\), with distinct thresholds
\(\tau_{\mathrm{num}},\tau_\sigma,\tau_\rho,\tau_v\) (and the refinement cap
\(\tau_\rho^{\mathrm{ref}}\) and wind-drift cap \(\epsilon_w\)) each predeclared or
calibrated from noise or physics and none tuned on source recovery
(Appendix~\ref{app:reporting_details}, Table~\ref{tab:thresholds}); it returns
source--basis coefficients and activities at the identifiable grouping, fitted
trajectories and residuals, uncertainty intervals, diagnostics, low-visibility
flags, and merge recommendations.  Algorithm~\ref{alg:iasa} states the procedure.

\begin{algorithm}[t]
\caption{Identifiability-Aware Source Apportionment}
\label{alg:iasa}
\footnotesize
\begin{algsteps}
  \item \textbf{Predeclare} the primary \(Q\), the physical lag-candidate grid,
  \(\Phi\), the inventory version, and the fixed-zero mask \(\mathcal F_0\),
  before inspecting any fit.
  \item \textbf{Wind:} convert \(\mathit{WD}/\mathit{WS}\) to transport vectors
  (Eq.~\ref{eq:wind_direction_conversion}), preserving masks/observed values for
  audit; query the adopted kernel imputer on each response-grid cell for
  \(\widehat{\mathbf w}_t(\mathbf x_g)\) and convert to grid displacement.
  \item \textbf{Response \& lag:} over unmasked source--basis pairs build
  \(\mathbf h_{t,kb}^{\mathrm{lag}}\) (Eq.~\ref{eq:constructed_response}); select
  the smallest \(L\) with \(\eta_L\le\tau_L\) from geometry alone
  (Eq.~\ref{eq:lag_convergence}), retaining the full lag sweep.
  \item \textbf{Align \& project:} apply the observed-PM\(_{2.5}\) selection
  identically to \(Y,H,Q\) and metadata (Eq.~\ref{eq:observation_selection}) with
  source-major, basis-minor columns; form \(\widetilde Y,\widetilde H_\Phi\)
  (Eq.~\ref{eq:projected_model}).
  \item \textbf{Fit:} solve \(\widehat{\mathbf c}\)
  (Eq.~\ref{eq:nnls_estimator}) and reconstruct
  \(\widehat\theta_k(t)=\sum_b\widehat c_{kb}\phi_b(t)\).
  \item \textbf{Diagnose:} compute \(r_{\mathrm{num}}\), the singular values,
  \(r_{\mathrm{eff}}(\tau_\sigma)\), visibility \(v_j\), absorption \(a_j\),
  coherence \(\rho_{ij}\), and ray distance \(d^{\mathrm{ray}}_{ij}\).
  \item \textbf{Flag:} weak set \(\mathcal W=\{j:v_j\le\tau_v\}\); eligible
  ambiguous pairs \(\mathcal A=\{(i,j):i,j\notin\mathcal W,\rho_{ij}>\tau_\rho\}\).
  \item \textbf{Group:} take deterministic connected components of the source
  ambiguity graph as conservative report groups
  (above); weak coefficients are flagged but create no
  edges, and group activities/contributions are member sums that do not replace
  the coefficient fit.
  \item \textbf{Adequacy:} if an externally calibrated noise model exists, run
  the refitted parametric-bootstrap residual test
  (below); otherwise label the residual summary
  uncalibrated.
  \item \textbf{Report} source estimates, observation/transport uncertainty,
  inventory robustness scenarios, residual adequacy, diagnostics, and
  conservative source groups.
\end{algsteps}
\end{algorithm}

The step order enforces the method's central discipline of separating fit from
diagnosis: all model choices, thresholds, and the fixed-zero mask are predeclared
(line~1), the response is built, projected, and fitted (lines~2--5), and only
afterwards are the identifiability diagnostics computed and conservative groups
reported (lines~6--8).  Because no threshold or column set is chosen after
inspecting the fitted coefficients (the lag itself is selected from response-matrix
geometry, not from \(\widehat{\mathbf c}\)), the reported identifiability cannot be
made artificially favorable by post-fit tuning.  The objective in
Equation~\ref{eq:nnls_estimator} is a smooth convex quadratic under a simple
nonnegativity constraint, solved with projected FISTA \citep{beck2009fast}
(Appendix~\ref{app:pytorch_solver}).

\noindent\textbf{Model Adequacy, Uncertainty, and Reporting.}
IASA attaches uncertainty to reported contributions from an active-set/ridge
covariance of the fitted coefficients and, when transport ensembles are supplied,
from empirical quantiles across ensemble refits.  Transport uncertainty and
inventory-robustness scenarios are never pooled into one interval.  Residual
magnitude alone is not calibrated evidence of misspecification.  When an externally
calibrated observation-noise model is available, IASA computes a
background-orthogonal residual statistic whose null distribution (not a plain
\(\chi^2\), since \(\widehat{\mathbf c}\) is fitted from the same data) is
calibrated by a parametric bootstrap (\(1000\) refits, \(\alpha=0.05\)).  Otherwise
it reports raw and projected residual summaries with
\texttt{calibration\_status=uncalibrated} and emits no pass/fail claim.  The test
is one-sided.  A rejection establishes a mismatch with the declared model but not
its cause, whereas a non-rejection does not establish inventory completeness, since
an omitted source whose signature lies in
\(\operatorname{span}[H_\Phi^{\mathrm{lag}},Q]\) is absorbed by fitted terms and
stays residual-invisible.  Prospective network adequacy over historical or
simulated wind windows, per-group reliability labeling, the per-sensor
contribution decomposition and footprints described above, and acceptance criteria for the
optional refinement are detailed in Appendix~\ref{app:reporting_details}.

\section{Evaluation}
\label{sec:evaluation}

We evaluate on a single New Delhi platform so the controlled and observed studies
share one regulatory sensor geometry, one PM\(_{2.5}\) and wind record, and one set
of proxy inventories.  We ask whether the geometry of \(\widetilde H_\Phi\) predicts
which coefficient, activity, and grouped-source claims are defensible.  A shared
base configuration is varied along one controlled axis at a time, avoiding
conflation of wind, source geometry, background flexibility, noise, and
forward-model error.  All runs use the \(40\times40\)
resolution, with a 72-hour transport window for the controlled experiments and
four consecutive one-week windows for the observed study.  Real imputed wind is a
controlled transport input when coefficients are synthetically known; the
observed-PM\(_{2.5}\) mode is a separate evaluation and is never assigned synthetic
source-recovery metrics.

\noindent\textbf{New Delhi Experimental Platform.}
The platform is built from hourly government PM\(_{2.5}\) and wind records
\citep{cpcb,cpcb_portal} spanning 1 May 2018--31 October 2020 over a 32-sensor
regulatory network.  PM\(_{2.5}\) is never imputed, and its valid-value mask defines
\(M_{\mathcal O}\) in Equation~\ref{eq:observation_selection}, with the same
ordered sensor--time rows retained in \(Y\), \(H_\Phi^{\mathrm{lag}}\), \(Q\), and
metadata.  Observed wind is converted (Equation~\ref{eq:wind_direction_conversion})
and completed on the response grid by the kernel coordinate-query imputer of
Section~\ref{subsec:wind_field_estimation}\label{subsec:new_delhi_wind}.  A learned
imputer was trained but not adopted (Section~\ref{subsec:results_wind_imputation}).
Four source groups (brick kilns, industries, population density, traffic) each
carry one normalized spatial proxy map and its own admissible temporal-basis
components (Table~\ref{tab:new_delhi_inventories},
Appendix~\ref{app:platform_details}), so fitted coefficients are in normalized
proxy units, not physical emission shares.  Both forward paths use the
open-boundary Gaussian puff operator of Section~\ref{subsec:plume_response}
evaluated at sensors, the matched generator for controlled recovery and the
inference operator for observed PM\(_{2.5}\).  An auxiliary advection--diffusion
simulator supplies a structural forward-model mismatch for the adequacy study
(Section~\ref{sec:identifiability_aware_method}), and normal backgrounds follow
Section~\ref{subsec:background_construction}.

The \emph{controlled} mode preserves this sensor geometry and transport while
assigning known nonnegative source--basis coefficients (synthetic or real gridded
wind, controlled observation noise, source splits, sensor-layout variants, and
optional operator mismatch), so it supports recovery and interval-coverage
metrics.  The \emph{observed} mode combines real PM\(_{2.5}\), its observation
mask, the imputed wind field, and the named normalized inventories; it has no
ground-truth source activities and is run over four consecutive one-week windows,
each fitted and reported independently.  Column-level records, Pusa-monitor
averaging and coverage, proxy normalization and crop, temporal-activity
assumptions, forward-model baselines, and simulator settings are given in
Appendix~\ref{app:platform_details}.\\

\noindent\textbf{Metrics and Reporting.}
Controlled trials, having known coefficients, report coefficient and
reconstructed-activity error, grouped activity/contribution error, projected
residual, and uncertainty-interval coverage.  All trials also report the singular
spectrum and its diagnostics (\(\sigma_J\), numerical and effective rank, condition
number, visibility, background absorption, eligible coherence and ray distance), the
weak set, triggering pairs, and deterministic report groups
(Section~\ref{sec:identifiability_theory}).  They also record lag convergence, the
fixed-zero mask, background specification, and ensemble provenance (transport and
inventory kept separate).  Observed New Delhi runs instead report raw and
projected residuals, fitted trajectories, normalized proxy coefficients and
contributions, uncertainty intervals, inventory-scenario robustness,
weak/ambiguous flags, conservative grouped contributions, and residual calibration
status.  A reported percentage is a fraction of the fitted inventory-attributed
sensor signal, not a physical-emission share.

\subsection{Controlled Results}
\label{subsec:controlled_results}

The controlled trials have known source--basis coefficients and therefore recovery
and coverage metrics.  The controlled geometries are deliberately small (mostly
\(K=2\) source groups) so that conditioning, coherence, and background absorption
can each be varied in isolation and read directly off \(\widetilde H_\Phi\).  The
observed study then exercises the full four-group, seven-coefficient pipeline on
real data.  Figure~\ref{fig:controlled} collects the baseline comparison and the
seven controlled diagnostics.  Full numeric tables are in
Appendix~\ref{app:expanded_results}.

\begin{figure*}[t]
\centering
\includegraphics[width=\textwidth]{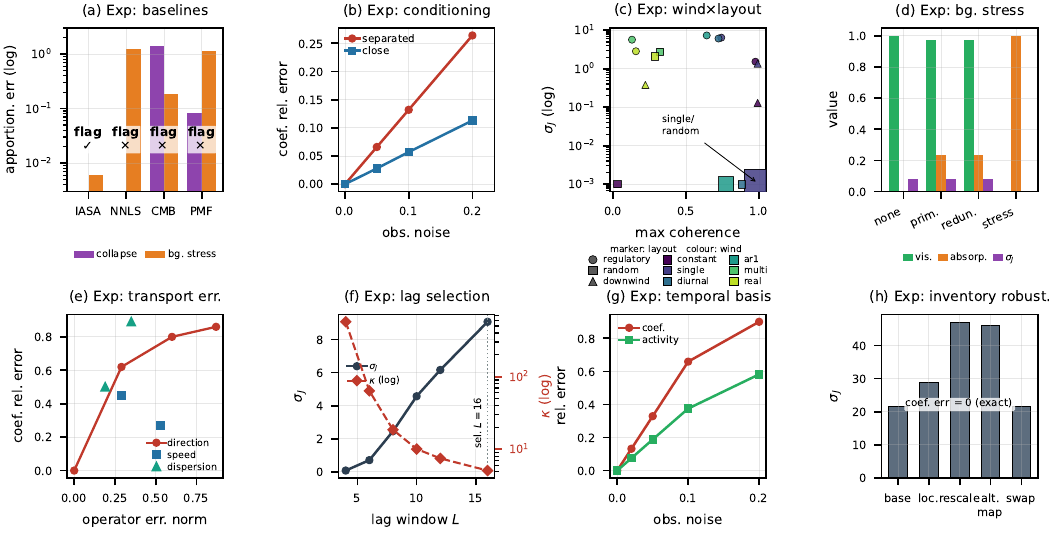}
\caption{Controlled experiments. \textbf{(a)}~IASA vs.\ three
identifiability-blind baselines (plain NNLS, chemical mass balance, PMF/NMF) on the
wind/geometry collapse and source-like background stress; bars are apportionment
share error (log) and the row below marks whether each method raised an
identifiability flag (IASA flags both; no baseline does).
\textbf{(b)}~Conditioning, not noise, sets the coefficient-recovery ceiling
(close $\sigma_J$9.6/$\kappa$2.3 vs.\ separated $\sigma_J$3.3/$\kappa$8.9).
\textbf{(c)}~$\sigma_J$ vs.\ maximum eligible coherence; colour~=~wind
provider (constant$\to$real, dark$\to$light), marker~=~layout (circle regulatory,
square random, triangle downwind), size~$\propto$~coefficient error; only
single/random collapses ($\sigma_J\!\to\!0$, error~$0.57$). \textbf{(d)}~A
source-like background drives absorption to~$1$ and $\sigma_J$/visibility to~$0$.
\textbf{(e)}~Coefficient error grows with the operator-error norm
(direction/dispersion axes; speed non-monotonic); structural PDE mismatch gives
adequacy rejection~$1.0$. \textbf{(f)}~As the lag window grows $\sigma_J$
rises and $\kappa$ falls; the rule selects $L{=}16$ from geometry alone.
\textbf{(g)}~Reconstructed activity error stays below coefficient error at
every noise level. \textbf{(h)}~$\sigma_J$ tracks the inventory version while
recovery stays exact. Values from Appendix~\ref{app:expanded_results}.}
\label{fig:controlled}
\end{figure*}

\noindent\textbf{Baselines.}
\label{subsec:results_baselines}
We compare IASA against three identifiability-blind baselines on the same data:
plain NNLS on the projected system (the identifiability layer ablated), chemical
mass balance on the unprojected response, and a PMF/NMF receptor factorization.  On
the wind/geometry collapse ($\text{coherence}\!\to\!1$, $\sigma_J\!\to\!0$) and the
source-like background stress ($\text{visibility}\!\to\!0$,
$\text{absorption}\!\to\!1$), IASA raises an identifiability flag in every scenario
and seed while none of the baselines do (Figure~\ref{fig:controlled}(a);
Table~\ref{tab:results_baselines}).
Under background stress IASA still recovers apportionment to $<\!0.01$ share error
whereas the baselines are off by $0.12$--$1.2$.  In the collapse the individual split
is non-identifiable (coefficient error $\approx0.57$ for the projected fit and
unbounded for the raw baselines), yet the baselines report it without warning.  A
standard method confidently reports a non-identifiable split; IASA reports only the
resolution its diagnostics support.

\noindent\textbf{Conditioning and Recovery.}
\label{subsec:results_h1}
We vary source geometry and Gaussian observation noise (0--20\% of the maximum
clean sensor signal).  Both geometries are full rank (numerical rank $=$ effective
rank $=2$), so the fixed geometry fixes $\sigma_J$ and $\kappa$ and coefficient
error grows linearly with noise.  The better-conditioned \emph{close} geometry
recovers proportionally more accurately than the \emph{separated} one
(Figure~\ref{fig:controlled}(b); Table~\ref{tab:results_h1}).  This is our central
claim: response-matrix geometry, not noise alone, sets the recovery ceiling.

\noindent\textbf{Wind Diversity and Sensor Geometry.}
\label{subsec:results_h4}
Matched source, basis, and sensor sets are evaluated across six wind providers
(constant, single-direction, diurnal, AR(1), multi-directional, and real gridded
New Delhi wind) and three sensor layouts (regulatory, seeded random,
downwind-focused), with identical source--basis columns on both sides of every
comparison (Figure~\ref{fig:controlled}(c); Table~\ref{tab:results_h4}).  Two
effects dominate: wind diversity improves conditioning (multi-directional and real
wind give the lowest maximum eligible coherence, while a single steady direction
pushes it toward one), and geometry matters as much as wind (under steady wind a
seeded random layout collapses $\sigma_J$ to $\approx0$ and produces the only
nonzero coefficient errors in the sweep, whereas the regulatory layout keeps
$\sigma_J$ well away from zero).  Separate historical and AR(1) window ensembles
both attain full rank with probability $1.0$ and no pairwise ambiguity
(Table~\ref{tab:results_h4_ens}).

\noindent\textbf{Background Stress.}
\label{subsec:results_h3}
We compare no background, the normal rank-four basis, a redundant-column basis with
the same span, and a labeled source-like stress basis, all declared before fitting.
The redundant basis is byte-identical to the primary in every reported quantity,
while the source-like \emph{stress} basis drives $\sigma_J$ and minimum visibility
to $0$ and absorption to $1$ (Figure~\ref{fig:controlled}(d);
Table~\ref{tab:results_h3}).  A too-flexible background can quietly absorb the
source-driven signal and destroy identifiability even when the fitted residual is
essentially unchanged.  Coefficient error is then misleading: it \emph{falls} in
the stress case precisely because the signal has been absorbed.

\noindent\textbf{Transport Error.}
Parametric perturbations of wind speed, direction, and dispersion within the puff
family propagate into attribution.  Coefficient error grows with the operator-error
norm along the direction (to $0.86$ at $20^\circ$) and dispersion (to $0.89$ at
$2\times$) axes, while the speed axis is non-monotonic, a reminder that a coarse
operator-error norm need not order attribution error (Figure~\ref{fig:controlled}(e);
Table~\ref{tab:results_h5a}).  A structural mismatch (open-boundary puff fitted to
edge-hold advection--diffusion observations, mismatch norm $0.44$) drives mean
coefficient error to $0.86$ and the refitted-bootstrap adequacy test rejects in
$100\%$ of trials, confirming power against forward-model misspecification, not only
missing sources.

\noindent\textbf{Lag-Window Selection.}
Physical travel and residence times define an increasing candidate lag grid.  As the
window grows the response stabilizes and conditioning improves monotonically
($\kappa$ falls from $569$ at $L{=}4$ to $5.1$ at $L{=}16$; $\sigma_J$ rises from
$0.08$ to $9.07$) while numerical rank and report-component count stay fixed and
recovery is exact (Figure~\ref{fig:controlled}(f); Table~\ref{tab:results_lag}).
The rule (Eq.~\ref{eq:lag_convergence}) picks the smallest lag whose relative
Frobenius change falls to $\tau_L=10^{-3}$ (here $L{=}16$) from geometry alone,
never from fitted coefficients.

\noindent\textbf{Temporal-Basis Recovery.}
Reconstructing known diurnal, intermittent, and day/night source--basis
coefficients at increasing noise separates coefficient error from
reconstructed-activity error.  The activity-trajectory error stays consistently below
the coefficient error ($0.58$ vs.\ $0.90$ at $20\%$ noise), at a near-constant ratio
$\approx0.57$ (Figure~\ref{fig:controlled}(g); Table~\ref{tab:results_temporal}),
because projecting noisy coefficients through the temporal bases damps isotropic
estimation noise more than the coherent signal.

\noindent\textbf{Inventory Robustness.}
Perturbing inventory locations, spatial scales, category assignments, and map
versions (each a separate scenario, never pooled with transport error) recovers the
known coefficients exactly, but the identifiability geometry moves with the
inventory version.  Rescaling or substituting the map nearly doubles $\sigma_J$ (to
$47.0$ and $46.0$ vs.\ $21.6$ baseline) while a symmetric category swap leaves it
unchanged (Figure~\ref{fig:controlled}(h); Table~\ref{tab:results_h5b}).  The
attribution \emph{target} thus changes with the assumed inventory even when recovery
is exact.  Three further controlled experiments not shown here (coherence and grouped
reporting, missing-source adequacy, and per-sensor footprints) are reported in
Appendix~\ref{app:additional_experiments}.

\subsection{Observed New Delhi}
\label{subsec:observed_results}

\begin{figure}[t]
\centering
\includegraphics[width=\columnwidth]{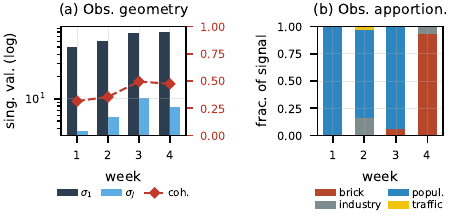}
\caption{Observed New Delhi. \textbf{(a)}~Weeks 1--4: every week is full rank; bars
$\sigma_1$/$\sigma_J$ (left, log) and max eligible coherence (red, right axis
$0$--$1$) (Table~\ref{tab:results_nd_ident}). \textbf{(b)}~proxy apportionment as a
fraction of fitted sensor signal (not physical emissions), swinging population- to
brick-kiln-dominated (Table~\ref{tab:results_nd_appt}).}
\label{fig:observed}
\end{figure}

\noindent\textbf{New Delhi Wind-Imputation Validation.}
\label{subsec:results_wind_imputation}
The learned wind model did not beat a non-spatial city-mean baseline on held-out
station--time vectors, so the spatially resolved kernel imputer is used throughout (the city mean lacks the
spatial resolution the response operator requires).  Dense wind truth is
unavailable, so gridded-field accuracy is assessed only through the controlled
synthetic-wind experiments, where the true wind is known.

\noindent\textbf{New Delhi Identifiability and Report Groups.}
\label{subsec:results_new_delhi_identifiability}
The observed study runs as four consecutive one-week windows (2018-05-01 through
2018-05-28), each a self-contained fit; results are not aggregated across windows. This window is chosen for its near-complete early-record coverage. It is a
demonstration window, the full 2018--2020 record was computationally infeasible.  Every week is full numerical rank ($7$ free coefficients after
the fixed-zero mask), finite-conditioned, with low maximum eligible coherence
($0.32$--$0.50$), an empty weak set, no cross-source ambiguous pairs, and four
singleton report components (Figure~\ref{fig:observed}(a);
Table~\ref{tab:results_nd_ident}).  Each source group is therefore individually
identifiable in every week and no conservative merge is required.

\noindent\textbf{New Delhi Proxy Apportionment and Uncertainty.}
\label{subsec:results_new_delhi_apportionment}
Figure~\ref{fig:observed}(b) reports each week's proxy apportionment as a
fraction of the fitted inventory-attributed sensor signal (the sum over groups of
the $L_1$ fitted per-group sensor-signal magnitude).  These are explicitly \emph{not}
physical-emission shares.  The mixture swings substantially, population-dominated in
weeks 1--3 and brick-kiln-dominated in week 4, which is why the windows are
reported separately rather than as one aggregate estimate.  The fit uses only the
$79$--$85\%$ of rows with a valid PM$_{2.5}$ measurement (no imputation), subtracting
each week's zero-source transported kriged initial-condition baseline (Pusa monitors
averaged) under the gridded kernel wind field. 
The week-4 brick-kiln
dominance (Figure~\ref{fig:observed}(b)) fits the end-of-season firing push before
the monsoon.

\section{Discussion and Conclusion}
\label{sec:discussion_conclusion}

We framed sparse-sensor apportionment as an identifiable-resolution problem:
inventories make attribution meaningful but need not make source groups separable.
Separability is governed by the projected response matrix \(\widetilde H_\Phi\).
Exact identifiability requires \(\operatorname{rank}(\widetilde H_\Phi)=J\), and
robust attribution its singular values, visibility, absorption, and coherence.
Prediction accuracy alone is therefore not evidence of identifiability.  IASA reports
contributions only where this geometry supports them and qualifies them otherwise,
and judges sensor, wind, and background design by fingerprint conditioning and
coherence, not coverage.  All claims are conditional on the
declared wind, transport, basis, inventory, and background.  The basis bounds
recoverable activity, aligned omitted sources stay residual-invisible, and
transport uncertainty (intervals) stays separate from inventory scenarios
(unpooled).  Richer meteorology, spatial diagnostics, and joint basis, transport,
and inventory estimation are natural next steps.

\section*{Acknowledgements}
Ankit Bhardwaj and Lakshminarayanan Subramanian were supported by the NSF Grant (Award Number 2335773) titled ``EAGER: Scalable Climate Modeling using Message-Passing Recurrent Neural Networks''. Lakshminarayanan Subramanian was also funded in part by the NSF Grant (award number OAC-2004572) titled ``A Data-informed Framework for the Representation of Sub-grid Scale Gravity Waves to Improve Climate Prediction''.


\appendix
\section{Detailed Identifiability Proofs}
\label{app:identifiability_proofs}

\noindent\textbf{Proof of Proposition~\ref{prop:rank_identifiability}.}
Let \(A=\widetilde H_\Phi\in\mathbb R^{N\times J}\).  If
\(\operatorname{rank}(A)=J\), then the rank-nullity theorem gives
\(\dim\ker(A)=J-\operatorname{rank}(A)=0\).  Hence
\(\ker(A)=\{\mathbf 0\}\).  If two feasible coefficient vectors
\(\mathbf c_1,\mathbf c_2\in\mathbb R_+^J\) produce the same projected
observation,
\[
A\mathbf c_1=A\mathbf c_2,
\]
then \(A(\mathbf c_1-\mathbf c_2)=0\).  Since the only vector in the kernel is
\(\mathbf 0\), we have \(\mathbf c_1-\mathbf c_2=\mathbf 0\), and therefore
\(\mathbf c_1=\mathbf c_2\).  Thus the coefficient vector is identifiable over
the nonnegative feasible set.

Conversely, suppose \(\operatorname{rank}(A)<J\).  By rank-nullity,
\(\dim\ker(A)>0\), so there exists a nonzero vector
\(\Delta\mathbf c\in\ker(A)\).  Choose any \(\beta>\|\Delta\mathbf c\|_\infty\)
and set \(\mathbf c=\beta\mathbf 1\).  Then \(\mathbf c\in\mathbb R_+^J\), and
\(\mathbf c+\Delta\mathbf c\in\mathbb R_+^J\) because every component satisfies
\[
\beta+\Delta c_j\ge \beta-\|\Delta\mathbf c\|_\infty>0.
\]
The two vectors are distinct because \(\Delta\mathbf c\ne0\), but they produce
the same projected observation:
\[
A(\mathbf c+\Delta\mathbf c)
=A\mathbf c+A\Delta\mathbf c
=A\mathbf c.
\]
Therefore uniform identifiability over \(\mathbb R_+^J\) fails whenever
\(\operatorname{rank}(A)<J\).  This proves the equivalence.

\noindent\textbf{Proof of Proposition~\ref{prop:noise_robust}.}
Again let \(A=\widetilde H_\Phi\), and assume \(A\) has full column rank
\(J\).  The projected observation model is
\[
\widetilde Y=A\mathbf c+\widetilde E.
\]
For unconstrained least squares,
\[
\widehat{\mathbf c}_{\mathrm{LS}}
=\arg\min_{\mathbf z}\|\widetilde Y-A\mathbf z\|_2^2
=A^\dagger\widetilde Y,
\]
where \(A^\dagger\) is the Moore--Penrose pseudoinverse.  Substituting the data
model gives
\[
\widehat{\mathbf c}_{\mathrm{LS}}
=A^\dagger(A\mathbf c+\widetilde E)
=\mathbf c+A^\dagger\widetilde E,
\]
because \(A^\dagger A=I_J\) for a full-column-rank matrix.  Hence
\[
\|\widehat{\mathbf c}_{\mathrm{LS}}-\mathbf c\|_2
=\|A^\dagger\widetilde E\|_2
\le \|A^\dagger\|_2\|\widetilde E\|_2.
\]
The singular values of \(A^\dagger\) are the reciprocals of the nonzero
singular values of \(A\), so \(\|A^\dagger\|_2=1/\sigma_J(A)\).  Therefore
\[
\|\widehat{\mathbf c}_{\mathrm{LS}}-\mathbf c\|_2
\le
\frac{\|\widetilde E\|_2}{\sigma_J(A)}.
\]

For nonnegative least squares, let \(\widehat{\mathbf c}\) minimize
\(\|\widetilde Y-A\mathbf z\|_2^2\) over \(\mathbf z\ge0\).  Since the true
\(\mathbf c\ge0\), it is feasible.  Optimality of \(\widehat{\mathbf c}\)
therefore implies
\[
\|\widetilde Y-A\widehat{\mathbf c}\|_2
\le
\|\widetilde Y-A\mathbf c\|_2
=\|\widetilde E\|_2.
\]
Let \(\Delta=\widehat{\mathbf c}-\mathbf c\).  Using
\(\widetilde Y=A\mathbf c+\widetilde E\),
\[
\widetilde Y-A\widehat{\mathbf c}
=\widetilde E-A\Delta.
\]
Thus
\[
\begin{aligned}
\|A\Delta\|_2
&=\|\widetilde E-(\widetilde E-A\Delta)\|_2\\
&\le
\|\widetilde E\|_2+\|\widetilde E-A\Delta\|_2
\le
2\|\widetilde E\|_2.
\end{aligned}
\]
Full column rank gives the lower singular-value inequality
\[
\|A\Delta\|_2\ge \sigma_J(A)\|\Delta\|_2.
\]
Combining the two inequalities yields
\[
\|\widehat{\mathbf c}-\mathbf c\|_2
=\|\Delta\|_2
\le
\frac{2\|\widetilde E\|_2}{\sigma_J(A)}.
\]
This proves the stated NNLS robustness bound.

\noindent\textbf{Ray distance for nonweak fingerprints
(Equation~\ref{eq:ray_distance_coherence_relation}).}
Fix nonweak projected fingerprints \(\widetilde{\mathbf h}_i,
\widetilde{\mathbf h}_j\), so \(\widetilde{\mathbf h}_j\ne0\).  The ray distance
minimizes
\[
g(\alpha)=\frac{\|\widetilde{\mathbf h}_i-\alpha\widetilde{\mathbf h}_j\|_2^2}
{\|\widetilde{\mathbf h}_i\|_2^2}
\]
over \(\alpha\in\mathbb R\).  This is strictly convex in \(\alpha\), and
\(g'(\alpha)=0\) gives
\(-2\widetilde{\mathbf h}_j^\top(\widetilde{\mathbf h}_i-\alpha
\widetilde{\mathbf h}_j)=0\), i.e.
\[
\alpha^\star=\frac{\widetilde{\mathbf h}_i^\top\widetilde{\mathbf h}_j}
{\|\widetilde{\mathbf h}_j\|_2^2}.
\]
Substituting,
\[
\|\widetilde{\mathbf h}_i-\alpha^\star\widetilde{\mathbf h}_j\|_2^2
=\|\widetilde{\mathbf h}_i\|_2^2
-\frac{(\widetilde{\mathbf h}_i^\top\widetilde{\mathbf h}_j)^2}
{\|\widetilde{\mathbf h}_j\|_2^2},
\]
and dividing by \(\|\widetilde{\mathbf h}_i\|_2^2\) yields
\((d^{\mathrm{ray}}_{ij})^2=1-\rho_{ij}^2\), so
\(d^{\mathrm{ray}}_{ij}=\sqrt{1-\rho_{ij}^2}\).  Minimizing over signed
\(\alpha\) makes the sign inside \(\rho_{ij}\) immaterial.

\noindent\textbf{Proof of Proposition~\ref{prop:operator_perturbation}.}
Let \(\widehat A=\widehat{\widetilde H}_\Phi\) have full column rank and let
\(A^\star=\widetilde H_\Phi^\star\) with \(\|\widehat A-A^\star\|_2\le\delta_H\).
The data obey
\(\widetilde Y=A^\star\mathbf c+\widetilde E
=\widehat A\mathbf c+[\widetilde E+(A^\star-\widehat A)\mathbf c]\).
The least-squares estimate using \(\widehat A\) is
\(\widehat{\mathbf c}=\widehat A^\dagger\widetilde Y\), and
\(\widehat A^\dagger\widehat A=I_J\), so
\[
\widehat{\mathbf c}-\mathbf c
=\widehat A^\dagger\bigl[\widetilde E+(A^\star-\widehat A)\mathbf c\bigr].
\]
Taking norms with \(\|\widehat A^\dagger\|_2=1/\sigma_J(\widehat A)\) and
\(\|(A^\star-\widehat A)\mathbf c\|_2\le\delta_H\|\mathbf c\|_2\) gives
\[
\|\widehat{\mathbf c}-\mathbf c\|_2
\le\frac{\|\widetilde E\|_2+\delta_H\|\mathbf c\|_2}{\sigma_J(\widehat A)},
\]
which is Equation~\ref{eq:operator_error_bound}.

\section{Identifiability Diagnostics: Definitions}
\label{app:diagnostics_definitions}

This appendix gives the full definitions of the diagnostics summarized in
Section~\ref{sec:identifiability_theory} and Table~\ref{tab:diagnostics}.

\begin{table}[t]
\centering
\small
\setlength{\tabcolsep}{3pt}
\caption{Identifiability diagnostics and reporting outputs produced with source
activity estimates.  The lower block lists the set- and graph-valued outputs that
accompany the scalar diagnostics.}
\label{tab:diagnostics}
\begin{tabularx}{\columnwidth}{l
  >{\raggedright\arraybackslash}X
  >{\raggedright\arraybackslash}X
  >{\raggedright\arraybackslash}X}
\toprule
Symbol & Name & Interpretation & Action \\
\midrule
\(r_{\mathrm{num}}(\widetilde H_\Phi)\) & numerical rank & exact separability & require rank \(J\) \\
\(\sigma_J(\widetilde H_\Phi)\) & padded min.\ singular value & noise robustness & flag zero/small \\
\(\kappa(\widetilde H_\Phi)\) & condition number & error amplification & \(\infty\) if deficient \\
\(r_{\mathrm{eff}}\) & effective rank & noise-level resolution & reduce grouping \\
\(v_j\) & coefficient visibility & measurable effect & flag weak \\
\(a_j\) & background absorption & signal lost to \(Q\) & warn after projection \\
\(\rho_{ij}\) & pairwise coherence & pairwise ambiguity & merge coherent \\
\(d^{\mathrm{ray}}_{ij}\) & ray distance & ray separation & N/A if weak \\
\midrule
\(\mathcal W\) & weak set & undetectable coefficients & flag, no edges \\
\(\mathcal A\) & ambiguous-pair set & confusable pairs & source edges \\
\(\mathcal E_{\mathrm{src}}\) & source ambiguity graph & source-level ambiguity & retain triggers \\
\(\mathcal G\) & connected report groups & identifiable resolution & report grouped \\
n/a & global unresolved warning & deficient, no edge & warn globally \\
\bottomrule
\end{tabularx}
\end{table}

For floating-point computation, the default numerical rank tolerance is
\(\tau_{\mathrm{num}}=\max(N,J)\,\epsilon_{\mathrm{mach}}\,\sigma_1\) with
numerical rank \(r_{\mathrm{num}}=\#\{i:\sigma_i>\tau_{\mathrm{num}}\}\),
where \(\epsilon_{\mathrm{mach}}\) is the machine precision of the SVD dtype
(about \(2.2\times10^{-16}\) in double precision).  This
\(\tau_{\mathrm{num}}\) is a purely numerical tolerance for detecting exact rank
deficiency; it is distinct from the scientific, noise-dependent threshold
\(\tau_\sigma\) below.  The primary identifiability score is the
\emph{padded minimum singular value} \(\sigma_J(\widetilde H_\Phi)\), where the
singular spectrum is padded with zeros whenever \(J>N\) or
\(\widetilde H_\Phi\) is rank deficient.  A rank-deficient response therefore
scores zero and is never reported as stable.
The condition number is \(\kappa(\widetilde H_\Phi)=\sigma_1/\sigma_J\) when
\(r_{\mathrm{num}}=J\) and \(\infty\) when \(r_{\mathrm{num}}<J\).
The smallest numerically positive value may additionally be reported as
\(\sigma_{\min,+}\), but it is not the identifiability score and never turns a
rank-deficient matrix into a stable one.  For a separately chosen,
noise-dependent threshold \(\tau_\sigma\), the effective rank is
\(r_{\mathrm{eff}}(\tau_\sigma)
=\#\{i:\sigma_i(\widetilde H_\Phi)>\tau_\sigma\}\).
Unlike \(\tau_{\mathrm{num}}\), the threshold \(\tau_\sigma\) is set from the
observation-noise scale, not from machine precision.  Given a per-row noise
standard deviation \(\sigma_e\) (from sensor calibration or a bootstrap residual
estimate), a singular direction is only resolvable when its singular value
exceeds the level at which iid noise projects onto it; a workable predeclared
default is \(\tau_\sigma=\sigma_e\sqrt{N}\).  It is recorded independently of
\(\tau_{\mathrm{num}}\) and is not tuned on source-recovery results.  If
\(r_{\mathrm{eff}}(\tau_\sigma)<J\), the sensor network does not support reliable
estimation of all \(J\) source--basis coefficients at that noise level.

For coefficient-specific diagnostics, let \(\widetilde{\mathbf h}_j\),
\(j=(k,b)\), be a projected source--basis fingerprint.  Its visibility and the
weak set at threshold \(\tau_v\ge0\) are
\(v_j=\|\widetilde{\mathbf h}_j\|_2\) and
\(\mathcal W=\{j:v_j\le\tau_v\}\).
A weak coefficient has little measurable effect after transport and background
correction: \(v_j\) is the sensor-time magnitude produced by one unit of
coefficient \(c_j\), and \(\tau_v\) marks fingerprints below the minimum
detectable signal.  A practical predeclared rule sets \(\tau_v\) from the noise
scale and a target signal-to-noise ratio, e.g.\ \(\tau_v=\mathrm{SNR}_{\min}\,
\sigma_e\), so a unit coefficient is flagged weak when even its full response is
indistinguishable from noise.  For eligible indices \(i,j\notin\mathcal W\),
pairwise ambiguity is measured by projected fingerprint coherence
\(\rho_{ij}=|\widetilde{\mathbf h}_i^\top\widetilde{\mathbf h}_j|/
(\|\widetilde{\mathbf h}_i\|_2\|\widetilde{\mathbf h}_j\|_2)\).
Values near one indicate nearly proportional sensor-time signatures.  The
ambiguous-pair set is
\(\mathcal A=\{(i,j):i,j\notin\mathcal W,\ \rho_{ij}>\tau_\rho\}\),
\(\tau_\rho\in[0,1]\).  The threshold \(\tau_\rho\) is predeclared or calibrated
on controlled experiments, not optimized after inspecting source recovery; a
conservative default takes \(\tau_\rho\) close to one so that only
near-proportional fingerprints are merged.  Coherence and ray distance are
reported only for eligible (nonweak) pairs; entries involving a weak fingerprint
are shown as ``N/A'' and excluded from extrema, while the weak-visibility flag
is always reported.

To express the same scale-invariant geometry as a distance, we also report the
ray distance for eligible, generally signed projected fingerprints,
\(d^{\mathrm{ray}}_{ij}
=\min_{\alpha\in\mathbb R}
\|\widetilde{\mathbf h}_i-\alpha\widetilde{\mathbf h}_j\|_2/
\|\widetilde{\mathbf h}_i\|_2\).
For nonweak fingerprints,
\begin{equation}
d^{\mathrm{ray}}_{ij}=\sqrt{1-\rho_{ij}^2},
\label{eq:ray_distance_coherence_relation}
\end{equation}
which we derive in Appendix~\ref{app:identifiability_proofs}; numerical
implementations evaluate \(\sqrt{\max(0,1-\rho_{ij}^2)}\).  This is an unoriented
linear-ray diagnostic, not a claim that projected fingerprints are nonnegative.
High coherence and small ray distance encode the same pairwise geometry, so ray
distance is a diagnostic geometrically equivalent to coherence under this
normalization, not an independent identifiability criterion.  Like coherence, it
is ``N/A'' for pairs involving \(\mathcal W\) and excluded from distance extrema.

\noindent\textbf{Background absorption.}
Background correction can improve robustness by removing broad temporal
variation, but it can also remove source-driven signal.  Let \(P_Q=QQ^\dagger\).
For coefficient fingerprint \(j=(k,b)\), define the background absorption
fraction \(a_j=\|P_Q\mathbf h_j^{\mathrm{lag}}\|_2/\|\mathbf h_j^{\mathrm{lag}}\|_2\).
If \(\|\mathbf h_j^{\mathrm{lag}}\|_2\) is zero at numerical tolerance,
\(a_j\) is undefined and reported as N/A; the coefficient remains explicitly
weak.
If \(a_j\approx 1\), most of the raw source--basis fingerprint lies in the
background space, and the projected visibility \(v_j\) will be small even if
the raw fingerprint is large.  This is why identifiability must be assessed
after projection.

\section{Temporal-Basis Notation and Diagnostic Conventions}
\label{app:temporal_basis_notation}

\begin{table}[h]
\centering
\small
\caption{Indices and objects in the temporal-basis apportionment model.}
\label{tab:temporal_basis_notation}
\begin{tabularx}{\columnwidth}{llX}
\toprule
Symbol & Shape or range & Meaning \\
\midrule
\(k\) & \(1,\ldots,K\) & source-group index \\
\(b\) & \(1,\ldots,B\) & temporal-basis index \\
\(j=(k,b)\) & \(1,\ldots,J\) & source-major, basis-minor coefficient index \\
\(\Phi\) & \(T\times B\) & nonnegative source-activity basis \\
\(C\) & \(K\times B\) & nonnegative source--basis coefficients \\
\(\mathbf c\) & \(J=KB\) & source-major vectorization of \(C\) \\
\(H_\Phi^{\mathrm{lag}}\) & \(N\times J\) & observed-row coefficient fingerprints \\
\(\widetilde H_\Phi\) & \(N\times J\) & background-projected fingerprints \\
\(\widetilde{\mathbf h}_{kb}\) & \(N\) & fingerprint for coefficient \(c_{kb}\) \\
\bottomrule
\end{tabularx}
\end{table}

The constant-activity model uses \(B=1\), \(\phi_1(t)=1\), and
\(c_{k1}=\theta_k\).  In the general model, diagnostics are computed for the
\(J\) coefficient fingerprints.  Source-level activity trajectories are
reconstructed as
\[
\widehat\theta_k(t)=\sum_{b=1}^{B}\widehat c_{kb}\phi_b(t);
\]
fingerprint columns are not summed to create an artificial source-level
diagnostic vector.

For a source-level ambiguity summary, IASA reports the largest eligible
coherence between coefficients belonging to different sources and retains the
source--basis pair that attains it.  Coherence and ray distance are undefined
for pairs involving a coefficient with \(v_j\le\tau_v\); tables display these
entries as N/A, exclude them from pairwise extrema, and report the weak
coefficient separately.

\section{Transport Response Construction and Fit Details}
\label{app:transport_response}

This appendix gives the construction, background, and estimation details deferred
from Section~\ref{sec:problem_setup}.

\noindent\textbf{Wind imputer query and unit conversion.}
The observed station vectors and masks are supplied to a coordinate-query wind
imputer: given sparse station observations
\(\{(\mathbf z_i,t,\mathbf u_{i,t},m_{i,t})\}\), it predicts transport vectors
at arbitrary spatial query locations and times.  For the IASA response operator,
we query the imputer on every response-grid cell \(\mathbf x_g\) and hour
\(t\), producing
\begin{equation}
\widehat{\mathbf w}_t(\mathbf x_g)
=
\widehat f_{\omega}\!\left(\mathbf x_g,t;
\{(\mathbf z_i,t',\mathbf u_{i,t'},m_{i,t'})\}\right),
\; g=1,\ldots,n.
\label{eq:fieldformer_wind_field}
\end{equation}
This yields a gridded wind field
\(\widehat W\in\mathbb R^{T\times n\times 2}\), rather than a station-only or
city-averaged sequence.  Real-data validation is performed by masking observed
station vectors and measuring held-out station-time error; dense city-wide wind
truth is not assumed.  Meteorological speed must be converted to grid
displacement before puff advection.  For response timestep \(\Delta t_{\mathrm{s}}\)
seconds and cell dimensions \(\Delta x_{\mathrm{m}},\Delta y_{\mathrm{m}}\), we use
\begin{equation}
v_x^{\mathrm{grid}}(\mathbf x_g,t)
=
\frac{\widehat U_x(\mathbf x_g,t)\Delta t_{\mathrm{s}}}{\Delta x_{\mathrm{m}}},
\qquad
v_y^{\mathrm{grid}}(\mathbf x_g,t)
=
\frac{\widehat V_y(\mathbf x_g,t)\Delta t_{\mathrm{s}}}{\Delta y_{\mathrm{m}}}.
\label{eq:wind_unit_conversion}
\end{equation}
Every response artifact records these scales, the coordinate convention, the
wind units, and the exact converted sequence.

\noindent\textbf{Wind-field ensembles.}
The imputed gridded wind field is an estimated transport input, not ground truth.
To propagate wind-imputation uncertainty, we construct wind-field ensembles
\(\{\widehat{\mathbf w}^{(r)}_t(\mathbf x_g)\}_{r=1}^R\) using held-out-calibrated
prediction error, station bootstrap resampling, checkpoint ensembles, or
perturbations matched to validation residuals.  Each ensemble member is queried
on the same response grid and converted to grid displacement before response
construction.  The resulting response matrices are tagged as transport
ensembles and are kept separate from inventory-robustness scenarios.

\noindent\textbf{Puff dynamics and dispersion.}
We instantiate the response with a Gaussian puff approximation.  For a
unit release from source cell \(i\) at location \(\mathbf r_i\) and release time
\(\tau\), the puff center is initialized as \(\mathbf z_i(\tau)=\mathbf r_i\) and
advected by the interpolated wind field,
\(d\mathbf z_i(a)/da=\widehat{\mathbf w}_a(\mathbf z_i(a))\), \(a\in[\tau,t]\).
Source cells have integer centers and domain extents
\([-0.5,N_x-0.5]\times[-0.5,N_y-0.5]\).  With substep size \(\delta a\), the
trajectory is discretized as
\(\mathbf z_i^{r+1}=\mathbf z_i^r+\delta a\,\widehat{\mathbf w}_{a_r}(\mathbf z_i^r)\).
The final substep is shortened so that the trajectory lands exactly at the
requested observation time.  If the center exits \(\Omega\), the whole
remaining puff is removed and never reflected, wrapped, clamped, or reinserted.
Otherwise its contribution is spread with an anisotropic Gaussian kernel
\begin{equation}
K_\psi(\mathbf x;\mathbf z_i,\Sigma_i)
=\frac{\exp\left(-\frac12\|\mathbf x-\mathbf z_i\|_{\Sigma_i^{-1}}^2\right)}
{2\pi |\Sigma_i|^{1/2}},
\label{eq:gaussian_kernel}
\end{equation}
where \(\psi\) denotes dispersion parameters.  A simple covariance model is
\begin{equation}
\Sigma_i(t,\tau)
= R_i(t,\tau)
\begin{bmatrix}
\sigma_\parallel^2(t-\tau) & 0\\
0 & \sigma_\perp^2(t-\tau)
\end{bmatrix}
R_i(t,\tau)^\top,
\label{eq:dispersion_covariance}
\end{equation}
with \(R_i\) aligned to the mean wind direction along the trajectory.  When the
mean wind norm is below the numerical threshold, the downwind direction is
undefined; the implementation uses the positive-\(x\) axis as a deterministic
tie-breaker for the anisotropic covariance orientation, recorded in the response
metadata and affecting only near-calm releases.  The age in this covariance is
lower-bounded by a positive minimum dispersion time, so same-time releases remain
finite: \(\text{effective age} = \max(t-\tau, t_{\min})\).

\noindent\textbf{Sensor evaluation and loss bookkeeping.}
Response-matrix entries are computed by evaluating the kernel directly at
sensor coordinates:
\begin{equation}
\left[O G_{t,\tau}^{\partial}
(\widehat{\mathbf w}_{\tau:t};\psi)\mathbf e_i\right]_s
=\chi_i(t,\tau)K_\psi(\mathbf x_s;\mathbf z_i(t),\Sigma_i(t,\tau)).
\label{eq:sensor_response}
\end{equation}
Here \(\chi_i(t,\tau)\) is an open-boundary survival indicator for a release
from source cell \(i\) at release time \(\tau\): it is one while the advected
kernel center remains inside the modeled domain and zero after the release has
exited.  Boundary and truncation losses are recorded only as diagnostic
quantities and are never used to renormalize the sensor response.  Boundary
loss denotes the portion of the Gaussian kernel lying outside the modeled
spatial domain, while truncation loss denotes the portion omitted by evaluating
retained mass only over a finite kernel support; both are estimated separately
by grid quadrature over the response grid.  These kernel losses are distinct
from whole-release exit loss, where the advected release center leaves the open
domain and no later sensor contribution is assigned to that release.  Because
the response matrix evaluates concentration at sensor points rather than
integrating over area-weighted grid cells, summing entries across sensors or
repeated sensor-time observations does not estimate retained pollutant mass.

\noindent\textbf{Baseline policy and differentiability.}
Response matrices are constructed for source-induced increments relative to a
declared initial-condition baseline.  The baseline policy is fixed before
source fitting, recorded with the response artifact, and applied consistently to
the observations and response rows.  Application-specific choices, such as
whether to subtract a zero-source transported initial field or to use a
zero-initial-state response, are part of the experimental instantiation rather
than the generic response definition.  The open-boundary lagged plume response is
implemented in PyTorch; tensor operations retained in the graph can be
differentiated, but discrete exit events and reporting logic are not assigned
gradients.

\noindent\textbf{Background basis and implicit projection.}
The normal background basis may contain a global constant, centered and scaled
elapsed-time polynomials, local-clock daily sine/cosine harmonics,
reference-coded day effects, standardized regional sensor-coordinate trends,
reference-coded sensor offsets, and a row-aligned user basis.  It depends only
on timestamps, day labels, sensor identity, and sensor coordinates, and its
effective rank is capped at eight.  Once the primary \(Q\) is fixed, its
coefficients are estimated implicitly by projection, or explicitly in the joint
fit; each alternative source-independent basis reports the resulting visibility,
absorption, rank, conditioning, and ambiguity components.  Letting
\(Q=U\Sigma V^\top\) be a thin singular value decomposition and retaining
columns \(U_r\) above the recorded rank tolerance, projection is applied
implicitly,
\begin{equation}
P_Q^\perp X=X-U_r(U_r^\top X),
\label{eq:implicit_background_projection}
\end{equation}
without constructing an \(N\times N\) matrix.  The row metadata for \(Y\),
\(H_\Phi^{\mathrm{lag}}\), and \(Q\) must be identical: each row must refer to the
same timestamp, sensor identity, and sensor index in the same order before
projection is applied.

\noindent\textbf{Regularizers and joint fit.}
The regularizer in Equation~\ref{eq:nnls_estimator} can be ridge stabilization
\(R(\mathbf c)=\|\mathbf c\|_2^2\) or a weak prior around inventory-derived
activity levels \(R(\mathbf c)=\|\mathbf c-\mathbf c_0\|_2^2\).  Equivalently, one
may fit background and source coefficients jointly,
\begin{equation}
(\widehat{\mathbf c},\widehat{\boldsymbol\gamma})
=\arg\min_{\mathbf c\ge 0,\boldsymbol\gamma}
\|Y-H_{\Phi}^{\mathrm{lag}}\mathbf c-Q\boldsymbol\gamma\|_2^2
+\lambda R(\mathbf c).
\label{eq:joint_fit}
\end{equation}
The projected formulation is used for identifiability because it exposes the
part of each source--basis fingerprint that remains after background correction.

\noindent\textbf{Constrained end-to-end refinement.}
IASA first constructs a fixed response matrix from the pretrained wind field and
default dispersion parameters, fits source activities, and computes
identifiability diagnostics for that declared response.  It then optionally
performs a constrained local refinement of the wind field, dispersion
parameters, source coefficients, and background coefficients, initialized at the
fixed-response solution and allowed only to make small physically constrained
corrections, so that improved sensor fit does not come from arbitrary changes to
transport.  Let \(\phi\) denote the wind-imputer parameters and \(\psi\) the
dispersion parameters of the puff response (such as \(\sigma_\parallel\),
\(\sigma_\perp\), and the minimum dispersion age), with \(\phi_0\) and \(\psi_0\)
their pretrained or default values.  We use \(R_{\mathrm{sm}}\) for a predeclared
wind-field smoothness penalty (temporal or spatial squared differences) and
\(\Psi_{\mathrm{phys}}\) for the physically admissible dispersion-parameter set.
The refinement objective is
\begin{align}
\mathcal L_{\mathrm{refine}}
=&\|Y-H_{\Phi}^{\mathrm{lag}}(\phi,\psi)\mathbf c-Q\boldsymbol\gamma\|_2^2
+\lambda_\theta R(\mathbf c) \nonumber\\
&+\lambda_w\|\widehat{\mathbf w}_{1:T}^{\phi}-\widehat{\mathbf w}_{1:T}^{\phi_0}\|_2^2
+\lambda_\psi\|\psi-\psi_0\|_2^2 \nonumber\\
&+\lambda_{\mathrm{sm}}R_{\mathrm{sm}}(\widehat{\mathbf w}^{\phi}),
\label{eq:refinement_objective}
\end{align}
subject to \(\mathbf c\ge 0\), \(\psi\in\Psi_{\mathrm{phys}}\), and
\(\|\widehat{\mathbf w}_{1:T}^{\phi}-\widehat{\mathbf w}_{1:T}^{\phi_0}\|_\infty
\le \epsilon_w\).  The refined response is accepted only as a constrained local
correction (acceptance criteria in Appendix~\ref{app:reporting_details}) and is
reported with its own response diagnostics.  Without these restrictions, wind and
plume parameters could fit sparse pollution sensors while weakening the physical
and source-level interpretation of \(\widehat{\mathbf c}\) and the reconstructed
activities.

\section{Response-Matrix Perturbations}
\label{app:operator_error_details}

This appendix expands the response-error discussion of
Section~\ref{sec:identifiability_theory}.  Write the projected
response discrepancy as
\(\Delta_H=\Delta_{\mathrm{tr}}+\Delta_{\mathrm{inv}}+\Delta_{\mathrm{int}}\),
where \(\Delta_{\mathrm{tr}}\) is the transport term, \(\Delta_{\mathrm{inv}}\)
the inventory term, and \(\Delta_{\mathrm{int}}\) their interaction.  The
decomposition is used to keep uncertainty honest: \(\Delta_{\mathrm{tr}}\) is
propagated through transport ensembles and may become an uncertainty interval,
whereas \(\Delta_{\mathrm{inv}}\) is reported only as named robustness scenarios
and is never pooled into an interval, and \(\Delta_{\mathrm{int}}\) is recorded
rather than pooled separately.  Let \(\widetilde H_\Phi^\star\) be the true
projected response and \(\widehat{\widetilde H}_\Phi\) the constructed one, with
\(\|\widehat{\widetilde H}_\Phi-\widetilde H_\Phi^\star\|_2\le\delta_H\).  Then
\(\widetilde Y=\widehat{\widetilde H}_\Phi\mathbf c
+[\widetilde E+(\widetilde H_\Phi^\star-\widehat{\widetilde H}_\Phi)\mathbf c]\),
so response error acts like additional observation error.

\begin{proposition}[Perturbation bound]
\label{prop:operator_perturbation}
Assume \(\widehat{\widetilde H}_\Phi\) has full column rank and
\(\|\widehat{\widetilde H}_\Phi-\widetilde H_\Phi^\star\|_2\le\delta_H\).
The unconstrained least-squares estimate using
\(\widehat{\widetilde H}_\Phi\) satisfies
\begin{equation}
\|\widehat{\mathbf c}-\mathbf c\|_2
\le
\frac{\|\widetilde E\|_2+\delta_H\|\mathbf c\|_2}
{\sigma_J(\widehat{\widetilde H}_\Phi)}.
\label{eq:operator_error_bound}
\end{equation}
\end{proposition}

A short proof, reusing the pseudoinverse bound of
Proposition~\ref{prop:noise_robust}, appears in
Appendix~\ref{app:identifiability_proofs}.
The same norm bound applies to all three discrepancy terms, but their
interpretations are not interchangeable.  Wind and dispersion ensembles
quantify transport uncertainty and may support uncertainty intervals under the
ensemble model.  Alternative inventory maps, locations, labels, or scales are
inventory robustness scenarios, not confidence intervals.  They are reported
separately, and source-specific attribution remains conditional on the supplied
inventory.  Either error type is most damaging when the projected response is
already ill-conditioned.

\section{Per-Sensor Source Footprints}
\label{app:footprint_details}

This appendix expands Section~\ref{sec:identifiability_theory}.  For an
observed row \((s,t)\), the background-corrected fit decomposes exactly as
\(\widetilde y_{s,t}=\sum_{k,b}\widetilde H_{\Phi,(s,t),(k,b)}\,\widehat c_{kb}\),
so the contribution of source \(k\) to sensor \(s\) over the record is
\(\widehat Y^{(s)}_k=\sum_t\sum_b
H_{\Phi,(s,t),(k,b)}^{\mathrm{lag}}\widehat c_{kb}\).  The projected form isolates
the identifiable part; the unprojected form gives the raw fitted sensor
contribution before background removal.  The spatial origin of a sensor's signal
is the pullback of its response row onto the source grid: for sensor \(s\) at time
\(t\), \(F_{s,t}(i)=\sum_{\ell\in\mathcal L_t}
[O G_{t,t-\ell}^{\partial}(\widehat{\mathbf w}_{t-\ell:t})]_{s,i}\ge 0\) and
\(F_s(i)=\sum_t F_{s,t}(i)\), a nonnegative field over source cells \(i\)
recording which upwind cells influence sensor \(s\).  Weighting by
\(\phi_b(t-\ell)\), the inventory map \(\mathbf s_k\), and \(\widehat c_{kb}\)
resolves the fitted contribution of each source group to each sensor by cell of
origin.  This is the puff response of Section~\ref{subsec:plume_response} read
backward from the sensor and requires no new transport operator.  Per-sensor
attribution uses only the rows of \(\widetilde H_\Phi\) for that sensor: for the
row submatrix \(\widetilde H_\Phi^{(s)}\),
\(\widetilde H_\Phi^{(s)\top}\widetilde H_\Phi^{(s)}\preceq
\widetilde H_\Phi^\top\widetilde H_\Phi\), so deleting rows cannot increase the
smallest singular value or the rank and a single sensor is never more
identifiable than the pooled network.  Footprints and per-sensor contributions
are therefore reported as spatial overlays of the global fit; per-sensor source
\emph{shares} are asserted only at the globally identifiable resolution,
aggregating to a report group wherever the pooled diagnostics require one, which
keeps the spatial view from implying separations the sensing system does not
support.

\section{Uncertainty, Adequacy, and Reporting Details}
\label{app:reporting_details}

This appendix expands the reporting components summarized in
Section~\ref{sec:identifiability_aware_method}, and gives the IASA thresholds
(Table~\ref{tab:thresholds}).

\begin{table}[t]
\centering
\small
\setlength{\tabcolsep}{3pt}
\caption{IASA thresholds, their meaning, and how they are set.  None is
optimized on source-recovery error.}
\label{tab:thresholds}
\begin{tabular}{lll}
\toprule
Threshold & Meaning & How set \\
\midrule
\(\tau_{\mathrm{num}}\) & numerical rank tol.\ & \(\max(N,J)\epsilon_{\mathrm{mach}}\sigma_1\) \\
\(\tau_\sigma\) & effective-rank floor & noise scale, e.g.\ \(\sigma_e\sqrt N\) \\
\(\tau_v\) & visibility floor & \(\mathrm{SNR}_{\min}\,\sigma_e\) \\
\(\tau_\rho\) & ambiguity coherence & predeclared, near \(1\) \\
\(\tau_\rho^{\mathrm{ref}}\) & refinement coherence & default \(\tau_\rho^{\mathrm{ref}}=\tau_\rho\) \\
\(\epsilon_w\) & wind-drift cap & physical wind tolerance \\
\bottomrule
\end{tabular}
\end{table}

\noindent\textbf{Uncertainty reporting.}
Let \(\widetilde r=\widetilde Y-\widetilde H_\Phi\widehat{\mathbf c}\) be the
projected residual and
\(\widehat\sigma^2=\|\widetilde r\|_2^2/\max(N-r_{\mathrm{eff}}(\tau_\sigma),1)\)
a residual variance estimate.  For an unconstrained or active-set approximation,
let \(\mathcal A_c=\{j:\widehat c_j>\tau_c\}\) and let
\(\widetilde H_{\Phi,\mathcal A_c}\) contain the active columns.  The covariance
of the active estimates is approximated by
\begin{equation}
\operatorname{Cov}(\widehat{\mathbf c}_{\mathcal A_c})
=
\widehat\sigma^2
(\widetilde H_{\Phi,\mathcal A_c}^\top
\widetilde H_{\Phi,\mathcal A_c}+\lambda I)^\dagger.
\label{eq:active_covariance}
\end{equation}
The use of \(N-r_{\mathrm{eff}}(\tau_\sigma)\) is a heuristic degrees-of-freedom
count that does not exactly account for the background rank \(r\) or the NNLS
active-set constraints.  Equation~\ref{eq:active_covariance} is the standard
active-set / ridge least-squares covariance approximation
\citep{seber2003linear}: its diagonal entries give per-coefficient uncertainty
and its off-diagonal entries expose coupled (jointly uncertain) estimates.  For
each transport ensemble member from
Section~\ref{subsec:wind_field_estimation}, IASA rebuilds
\(\widetilde H_\Phi^{(b)}\), refits \(\widehat{\mathbf c}^{(b)}\), and reports
empirical quantiles across members as transport uncertainty.  Every ensemble is
tagged by provenance: wind and dispersion members have
\texttt{ensemble\_kind=transport} and may form transport-uncertainty intervals,
whereas alternative inventory maps, source locations, labels, or scales have
\texttt{ensemble\_kind=inventory} and produce scenario-wise robustness tables.
The two kinds are never pooled into one interval, and active-set intervals
describe the fitted observation model without converting inventory scenarios into
probabilistic uncertainty.

\noindent\textbf{Residual model-adequacy check.}
Let \(Z\) be an orthonormal basis for the background-orthogonal observed-row
space, and define the nondegenerate residual coordinates and externally supplied
noise covariance
\begin{equation}
\bar r=Z^\top(Y-H_\Phi^{\mathrm{lag}}\widehat{\mathbf c}),
\qquad
\bar\Sigma_e=Z^\top\Sigma_e Z,
\label{eq:adequacy_coordinates}
\end{equation}
where \(\Sigma_e\) is an externally declared observation-noise covariance and
\(Z\) removes the background directions so that \(\bar r\) carries no degrees of
freedom absorbed by \(Q\).  For independent homoscedastic noise
\(\Sigma_e=\sigma_e^2 I\), orthonormality of \(Z\) gives
\(\bar\Sigma_e=\sigma_e^2 I\) and \(T_{\mathrm{res}}=\|\bar r\|_2^2/\sigma_e^2\).
When \(\bar\Sigma_e\) is calibrated independently of this fitted residual and is
positive definite, the adequacy statistic is
\begin{equation}
T_{\mathrm{res}}=\bar r^\top\bar\Sigma_e^{-1}\bar r.
\label{eq:residual_adequacy_statistic}
\end{equation}
Its null distribution is not a plain \(\chi^2\), because \(\widehat{\mathbf c}\)
is fitted from the same data; IASA therefore calibrates it by parametric
bootstrap.  It simulates \(B_{\mathrm{res}}=1000\) synthetic observation sets from
the fitted source/background model plus the declared noise model, refits the
coefficients on each draw, and recomputes \(T_{\mathrm{res}}\) to form the null
distribution.  At \(\alpha=0.05\) it flags inadequacy when the observed statistic
exceeds the empirical \((1-\alpha)\)-quantile; the Monte Carlo \(p\)-value uses
the standard add-one correction.  If no externally calibrated noise model is
supplied, IASA reports raw and projected residual norms and sensor-wise,
time-wise, and autocorrelation summaries with
\texttt{calibration\_status=uncalibrated} and emits no pass/fail adequacy claim.
Rejection establishes a mismatch with the declared model but does not identify a
missing source, and non-rejection does not establish inventory completeness: an
omitted source whose sensor-time signature lies in
\(\operatorname{span}[H_\Phi^{\mathrm{lag}},Q]\) can be absorbed by fitted terms
and remain residual-invisible.

\noindent\textbf{Wind-distribution diagnostics.}
Diagnostics for one realized wind window answer a retrospective question.
Prospective network adequacy is assessed empirically over predeclared historical
or simulated wind windows.  Across those response matrices, IASA reports the
5th, 50th, and 95th percentiles of \(\sigma_J\), the probability of full
numerical and effective rank, the probability that each coefficient is weak,
the probability that each source pair is ambiguous, and the frequency of each
conservative report component.  These are Monte Carlo or empirical
distributional summaries under the sampled wind population, not a new
identifiability theorem.

\noindent\textbf{Acceptance criteria for refinement.}
If the optional constrained refinement of
Appendix~\ref{app:transport_response} is used, it is accepted only if it improves
fit without degrading separability.  Let \(\widetilde H_{\Phi,0}\) be the
projected response before refinement and \(\widetilde H_{\Phi,\mathrm{ref}}\)
after refinement.  We require
\(\sigma_J(\widetilde H_{\Phi,\mathrm{ref}})\ge
(1-\eta_{\mathrm{id}})\sigma_J(\widetilde H_{\Phi,0})\) and
\(\max_{i\ne j,\,i,j\notin\mathcal W}\rho_{ij}^{\mathrm{ref}}\le
\tau_\rho^{\mathrm{ref}}\).  These checks prevent end-to-end tuning from improving
sensor fit by making source--basis fingerprints less distinguishable.  If either
response has fewer than \(J\) numerically nonzero singular values, its
\(\sigma_J\) is taken to be zero
(Section~\ref{sec:identifiability_theory}), so refinement cannot be
accepted by exploiting a rank-deficient response.

\noindent\textbf{Reporting protocol.}
The source ambiguity graph, its deterministic connected components, and the
transitive over-merging of an \(A\)--\(B\)--\(C\) chain are defined in
Section~\ref{sec:identifiability_theory}.  Operationally, those components are the
conservative recommended report groups, each edge retains its triggering
coefficient pair, and rank deficiency without an eligible pair yields a global
unresolved warning rather than an invented edge.  For each reported source group
\(G_a\), IASA reports the activity trajectory
\(\widehat{\boldsymbol\theta}_{G_a,1:T}\), its interval
\(\mathrm{CI}(\widehat\theta_{G_a})\), member visibilities
\(\{v_{kb}:k\in G_a\}\), and the maximum eligible cross-group coherence
\(\max_{k\in G_a,\,k'\notin G_a,\,(k,b),(k',b')\notin\mathcal W}
\rho_{(k,b),(k',b')}\).  A source contribution is marked reliable only when its
constituent reported coefficients are above the visibility threshold, its
uncertainty interval is sufficiently narrow, and its eligible coherence with
other reported groups is below the ambiguity threshold.  Reports retain the basis
names of every weak coefficient and the coefficient pair attaining each
cross-source coherence maximum; coherence and ray distance involving weak
fingerprints are displayed as ``N/A'' and excluded from extrema, while the weak
flags remain visible.  Otherwise IASA labels the source as weakly visible,
ambiguous with another source group, or merged into a coarser identifiable group.

\section{New Delhi Platform Details}
\label{app:platform_details}

This appendix expands the platform construction of
Section~\ref{sec:evaluation}: the source inventories
(Table~\ref{tab:new_delhi_inventories}), the kriged initial-condition baseline,
and the auxiliary simulator.

\begin{table}[t]
\centering
\small
\setlength{\tabcolsep}{4pt}
\caption{New Delhi source groups, spatial proxies, and temporal-basis components.}
\label{tab:new_delhi_inventories}
\begin{tabularx}{\columnwidth}{l
  >{\raggedright\arraybackslash}X
  >{\raggedright\arraybackslash}X}
\toprule
Source group & Spatial proxy & Temporal-basis components \\
\midrule
Brick kilns & regional kiln map & alternating 12-hour blocks \\
Industries & regional industry map & continuous + daytime (07--19) \\
Population density & gridded population & peaks 07:00, 13:00, 19:00 \\
Traffic & road-network map & diurnal slots 00/06/12/18 \\
\bottomrule
\end{tabularx}
\end{table}

\noindent\textbf{Observations, sensors, and inventory preparation.}
The hourly government records \citep{cpcb,cpcb_portal} (columns
\texttt{monitor\_id}, \texttt{timestamp\_round}, \texttt{AT}, \texttt{RH},
\texttt{WD}, \texttt{WS}, \texttt{pm10}, \texttt{pm25}) span 21,960 hourly
timestamps from 1 May 2018 through 31 October 2020 in Indian Standard Time;
approximately 74.9\% of wind-direction, 72.8\% of wind-speed, and 90.5\% of
PM\(_{2.5}\) entries are observed, with a valid-value mask retained alongside
station coordinates and timestamps.  The two Pusa monitors (\texttt{Pusa\_IMD} and
\texttt{Pusa\_DPCC}) are averaged by timestamp into \texttt{Pusa\_averaged} (mean
coordinates, per-variable hourly averaging), giving the final 32-sensor layout;
PM\(_{2.5}\) is never imputed and its flattened valid-value mask defines
\(M_{\mathcal O}\), with the same ordered sensor--time rows retained in \(Y\),
\(H_\Phi^{\mathrm{lag}}\), \(Q\), and metadata.  Brick-kiln and industry proxies
follow the regional emissions inventory \citep{GUTTIKUNDA2013101}, population
density the gridded population product \citep{ColumbiaUniversity2018}, and the
traffic map is derived from road map data \citep{google_maps}.  Every
\(80\times80\) map is cropped to the New Delhi study window (rows \(21{:}61\),
columns \(16{:}56\)) covering the regulatory network, giving a \(40\times40\) map,
then divided by its own cropped 99th percentile.  Traffic is modeled as a single
road-network source whose nearest-slot maps at 00:00, 06:00, 12:00, and 18:00
define traffic-specific temporal-basis components under the shared dictionary
\(\Phi\) and the fixed-zero mask \(\mathcal F_0\) (assuming a time-invariant
congestion spatial pattern whose amplitude varies); brick-kiln activity uses
deterministic alternating 12-hour blocks, industry a continuous operating fraction
with higher daytime (07:00--19:00) activity, and population a baseline with smooth
cooking peaks.  These are deterministic study assumptions, not observed emissions
truth, and each map's independent proxy normalization is why fitted coefficients
and contributions are in normalized proxy units rather than physical emission
totals or shares.

\noindent\textbf{Forward-model baselines and study modes.}
Both forward paths use the New Delhi-derived \(40\times40\) domain and regulatory
sensor coordinates, with the open-boundary Gaussian puff operator evaluated
directly at sensors.  The same operator is the matched generator for controlled
recovery experiments and the inference operator for observed PM\(_{2.5}\); only the
initial-condition baseline differs by mode: a zero-source, zero-initial-state
baseline for controlled operator-matched runs, and a subtracted zero-source
transported kriged-initial-condition baseline for observed runs.  Normal
backgrounds follow Section~\ref{subsec:background_construction}, with the base
controlled background the rank-four member.  Where a study is explicitly about
inventory geometry (the coherence and inventory-robustness experiments) the real
New Delhi inventory maps are used; the remaining controlled experiments use compact
synthetic source cells for tractability, since their conclusions concern
conditioning, coherence, and recovery rather than specific inventory footprints.

\noindent\textbf{Kriged initial-condition baseline.}
For observed New Delhi runs, the initial pollution field is estimated from the
first available regulatory observations (the two Pusa monitors averaged) by a
normalized Gaussian-kernel spatial interpolation on the response grid, a kriging
surrogate that avoids an additional heavyweight geostatistics dependency.  Its
zero-source transported trajectory, obtained by propagating this initial field as
a single release through the same open-boundary puff operator (its Gaussian kernel
is the advection--diffusion Green's function, so units are preserved), is
subtracted from the sensor observations before fitting source activities.  Source
coefficients are therefore fitted to the residual sensor-time signal after
removing the contribution of the initialized background field and the low-rank
temporal background basis.  Kernel interpolation from sparse first-hour
observations is itself an ill-posed spatial estimate, so this baseline injects
some uncertainty; the effect is limited by declaring and subtracting the baseline
before fitting and by the low-rank background basis \(Q\), which absorbs residual
broad structure.  Controlled operator-matched experiments may instead use a zero
initial field when the synthetic data are generated from the same
zero-initial-state response.

\noindent\textbf{Auxiliary advection--diffusion simulator.}
The auxiliary advection--diffusion simulator uses first-order upwind advection, a
five-point Laplacian, a two-stage Heun update, diffusivity \(3\times10^{-4}\), and
edge-hold boundaries.  It provides a structural forward-model mismatch: data
generated by this distinct transport operator are fit with the open-boundary puff
response, testing the puff approximation against a different operator family.
This mismatch also drives the residual-adequacy test of
Section~\ref{sec:identifiability_aware_method}, letting us check its power against structural
transport misspecification and not only against missing sources.

\section{Expanded Results}
\label{app:expanded_results}

Full numeric tables backing the Section~\ref{sec:evaluation} figures: the seven
controlled diagnostics and the baseline comparison of Figure~\ref{fig:controlled},
and the observed New Delhi weeks of Figure~\ref{fig:observed}.

\subsection{Conditioning (Exp 1)}
Coefficient recovery versus conditioning and noise; the better-conditioned geometry
recovers more accurately at matched noise.
\begin{table}[t]
\centering
\footnotesize
\caption{Experiment 1: coefficient recovery versus conditioning and noise. The
better-conditioned geometry recovers more accurately at matched noise.}
\label{tab:results_h1}
\begin{tabular}{llrrrr}
\toprule
geom. & noise & $\sigma_J$ & $\kappa$ & coef.\ err. & resid. \\
\midrule
separated & 0.00 & 3.27 & 8.89 & 0.000 & 0.00 \\
separated & 0.05 & 3.27 & 8.89 & 0.066 & 6.37 \\
separated & 0.10 & 3.27 & 8.89 & 0.132 & 12.73 \\
separated & 0.20 & 3.27 & 8.89 & 0.264 & 25.46 \\
close     & 0.00 & 9.60 & 2.26 & 0.000 & 0.00 \\
close     & 0.05 & 9.60 & 2.26 & 0.028 & 8.42 \\
close     & 0.10 & 9.60 & 2.26 & 0.057 & 16.84 \\
close     & 0.20 & 9.60 & 2.26 & 0.113 & 33.68 \\
\bottomrule
\end{tabular}
\end{table}

\subsection{Wind Diversity and Sensor Geometry (Exp 4)}
$\sigma_J$, maximum eligible coherence, and coefficient error across wind provider
$\times$ sensor layout, and the wind-window ensemble spread.
\begin{table}[t]
\centering
\footnotesize
\caption{Experiment 4: $\sigma_J$, maximum eligible coherence, and coefficient
error across wind provider $\times$ sensor layout.}
\label{tab:results_h4}
\begin{tabular}{llrrr}
\toprule
wind & layout & $\sigma_J$ & max.\ coh. & coef.\ err. \\
\midrule
constant & regulatory & 1.53 & 0.975 & 0.000 \\
constant & random     & 0.00 & 0.028 & 0.000 \\
constant & downwind   & 0.13 & 0.990 & 0.000 \\
single   & regulatory & 6.44 & 0.741 & 0.000 \\
single   & random     & 0.00 & 1.000 & 0.573 \\
single   & downwind   & 1.36 & 0.990 & 0.000 \\
diurnal  & regulatory & 6.11 & 0.722 & 0.000 \\
diurnal  & random     & 0.00 & 0.880 & 0.001 \\
ar1      & regulatory & 7.35 & 0.641 & 0.000 \\
ar1      & random     & 0.00 & 0.774 & 0.115 \\
multi    & regulatory & 5.70 & 0.129 & 0.000 \\
multi    & random     & 2.79 & 0.322 & 0.000 \\
real     & regulatory & 2.85 & 0.156 & 0.000 \\
real     & random     & 2.10 & 0.287 & 0.000 \\
real     & downwind   & 0.38 & 0.221 & 0.000 \\
\bottomrule
\end{tabular}
\end{table}
\begin{table}[t]
\centering
\footnotesize
\caption{Experiment 4: wind-window ensemble distribution ($\sigma_J$ percentiles;
both ensembles reach full rank with probability 1 and show no ambiguity).}
\label{tab:results_h4_ens}
\begin{tabular}{lrrr}
\toprule
ensemble & $\sigma_J^{5}$ & $\sigma_J^{50}$ & $\sigma_J^{95}$ \\
\midrule
historical (real slices) & 1.45 & 4.51 & 9.96 \\
simulated (AR(1))        & 3.50 & 5.58 & 6.80 \\
\bottomrule
\end{tabular}
\end{table}

\subsection{Background Stress (Exp 3)}
Redundant $\equiv$ primary (same span); the source-like stress basis collapses
identifiability ($\sigma_J\!\to\!0$, absorption $\to\!1$).
\begin{table}[t]
\centering
\footnotesize
\caption{Experiment 3: background stress. Redundant $\equiv$ primary (same span);
the source-like stress basis collapses identifiability ($\sigma_J\!\to\!0$,
absorption $\to\!1$).}
\label{tab:results_h3}
\begin{tabular}{lrrrr}
\toprule
background & min.\ vis. & absorp. & $\sigma_J$ & coef.\ err. \\
\midrule
none      & 1.000 & 0.000 & 0.080 & 2.36 \\
primary   & 0.972 & 0.235 & 0.080 & 2.61 \\
redundant & 0.972 & 0.235 & 0.080 & 2.61 \\
stress    & 0.000 & 1.000 & 0.000 & 0.78 \\
\bottomrule
\end{tabular}
\end{table}

\subsection{Transport Error (Exp 5)}
\label{subsec:results_h5a}
Parametric perturbations (wind speed, direction, dispersion); a structural PDE
mismatch generated with the auxiliary edge-hold advection--diffusion solver yields
adequacy rejection rate $1.0$.
\begin{table}[t]
\centering
\footnotesize
\caption{Experiment 5: parametric transport error. Coefficient error rises with
the operator error norm on the direction and dispersion axes (the speed axis is
non-monotonic). Structural PDE mismatch: adequacy rejection rate $1.0$.}
\label{tab:results_h5a}
\begin{tabular}{llrrr}
\toprule
axis & value & op.\ err. & $\sigma_J$ & coef.\ err. \\
\midrule
direction ($^\circ$) & 0  & 0.00 & 0.87 & 0.00 \\
direction ($^\circ$) & 5  & 0.29 & 3.37 & 0.62 \\
direction ($^\circ$) & 10 & 0.60 & 7.07 & 0.80 \\
direction ($^\circ$) & 20 & 0.87 & 7.98 & 0.86 \\
speed ($\times$)     & 1.25 & 0.29 & 2.09 & 0.45 \\
speed ($\times$)     & 1.50 & 0.53 & 4.57 & 0.27 \\
dispersion ($\times$)& 1.50 & 0.19 & 1.42 & 0.50 \\
dispersion ($\times$)& 2.00 & 0.35 & 1.76 & 0.89 \\
\bottomrule
\end{tabular}
\end{table}

\subsection{Inventory Robustness (Exp 6)}
\label{subsec:results_h5b}
Location, scale, category, and map-version perturbations, each a separate scenario
never pooled with transport error; recovery exact, $\sigma_J$ tracks the version.
\begin{table}[t]
\centering
\footnotesize
\caption{Experiment 6: inventory robustness on the real maps. Recovery is exact in
every scenario; $\sigma_J$ tracks the inventory version.}
\label{tab:results_h5b}
\begin{tabular}{lrr}
\toprule
scenario & $\sigma_J$ & coef.\ err. \\
\midrule
baseline        & 21.56 & 0.0 \\
location shift   & 28.85 & 0.0 \\
spatial rescale  & 47.00 & 0.0 \\
alt.\ map version & 46.01 & 0.0 \\
category swap     & 21.56 & 0.0 \\
\bottomrule
\end{tabular}
\end{table}

\subsection{Lag-Window Selection (Exp 7)}
\label{subsec:results_lag_selection}
Conditioning improves with lag; the selected $L=16$ is chosen from
Equation~\ref{eq:lag_convergence} alone, not fitted coefficients.
\begin{table}[t]
\centering
\footnotesize
\caption{Experiment 7: lag-window selection. Conditioning improves with lag; the
selected lag ($L=16$) is chosen from geometry alone, not fitted coefficients.}
\label{tab:results_lag}
\begin{tabular}{rrrrr}
\toprule
lag $L$ & $\sigma_J$ & $\kappa$ & \#comp. & coef.\ err. \\
\midrule
4  & 0.08 & 569.2 & 2 & 0.0 \\
6  & 0.71 & 63.8  & 2 & 0.0 \\
8  & 2.46 & 18.6  & 2 & 0.0 \\
10 & 4.57 & 10.1  & 2 & 0.0 \\
12 & 6.16 & 7.5   & 2 & 0.0 \\
16 & 9.07 & 5.1   & 2 & 0.0 \\
\bottomrule
\end{tabular}
\end{table}

\subsection{Temporal-Basis Recovery (Exp 9)}
\label{subsec:results_temporal_basis}
Activity-trajectory error stays below coefficient error at every noise level (fixed
$\sigma_J$), a fixed geometric property of the temporal-basis map.
\begin{table}[t]
\centering
\footnotesize
\caption{Experiment 9: temporal-basis recovery (bases: diurnal, block, day/night).
Activity-trajectory error stays below coefficient error at every noise level.}
\label{tab:results_temporal}
\begin{tabular}{rrr}
\toprule
noise & coef.\ err. & activity err. \\
\midrule
0.00 & 0.000 & 0.000 \\
0.02 & 0.132 & 0.075 \\
0.05 & 0.329 & 0.188 \\
0.10 & 0.659 & 0.375 \\
0.20 & 0.900 & 0.582 \\
\bottomrule
\end{tabular}
\end{table}

\subsection{Baseline Comparison (Exp 11)}
IASA against plain NNLS (identifiability layer ablated), chemical mass balance
(unprojected), and a PMF/NMF receptor factorization, on the wind/geometry collapse
and source-like background stress. IASA flags both and, under stress, recovers
apportionment to $<\!0.01$ share error; no baseline flags either case.
\begin{table}[t]
\centering
\footnotesize
\caption{Experiment 11: IASA vs.\ identifiability-blind baselines on two
non-identifiable scenarios (seed 0). Apportionment share error and whether an
identifiability flag was raised. IASA flags in both scenarios; no baseline does.}
\label{tab:results_baselines}
\begin{tabular}{lllr}
\toprule
scenario & method & flag & share err. \\
\midrule
collapse    & IASA       & yes & 0.000 \\
collapse    & plain NNLS & no  & 0.000 \\
collapse    & CMB        & no  & 1.414 \\
collapse    & PMF        & no  & 0.085 \\
\midrule
bg.\ stress & IASA       & yes & 0.006 \\
bg.\ stress & plain NNLS & no  & 1.231 \\
bg.\ stress & CMB        & no  & 0.189 \\
bg.\ stress & PMF        & no  & 1.152 \\
\bottomrule
\end{tabular}
\end{table}

\subsection{Observed New Delhi (weeks 1--4)}
Per-week identifiability geometry, proxy apportionment (fraction of fitted sensor
signal, not physical emissions), and sensor-fit/residual diagnostics (PM$_{2.5}$
never imputed; adequacy uncalibrated).
\begin{table}[t]
\centering
\footnotesize
\caption{Observed New Delhi identifiability, weeks 1--4. All weeks are full rank,
finite-conditioned, with singleton report groups.}
\label{tab:results_nd_ident}
\begin{tabular}{lrrrrl}
\toprule
week & $\sigma_1$ & $\sigma_J$ & num.\ rank & max.\ coh. & groups \\
\midrule
1 & 50.1 & 3.71  & 7 & 0.316 & 4 singletons \\
2 & 60.3 & 5.67  & 7 & 0.355 & 4 singletons \\
3 & 77.0 & 10.15 & 7 & 0.496 & 4 singletons \\
4 & 81.2 & 7.65  & 7 & 0.475 & 4 singletons \\
\bottomrule
\end{tabular}
\end{table}
\begin{table}[t]
\centering
\footnotesize
\caption{Observed New Delhi proxy apportionment, weeks 1--4 (fraction of fitted
sensor signal, not physical emissions). Note the large week-to-week swing.}
\label{tab:results_nd_appt}
\begin{tabular}{lrrrr}
\toprule
week & brick kilns & industries & population & traffic \\
\midrule
1 & 0.00 & 0.00 & 1.00 & 0.00 \\
2 & 0.00 & 0.16 & 0.81 & 0.04 \\
3 & 0.06 & 0.00 & 0.94 & 0.00 \\
4 & 0.93 & 0.07 & 0.00 & 0.00 \\
\bottomrule
\end{tabular}
\end{table}
\begin{table}[t]
\centering
\footnotesize
\caption{Observed New Delhi sensor fit and residual diagnostics, weeks 1--4.
PM$_{2.5}$ never imputed; adequacy uncalibrated (no verdict).}
\label{tab:results_nd_resid}
\begin{tabular}{lrrrl}
\toprule
week & obs.\ rows & mask frac. & resid.\ norm & adequacy \\
\midrule
1 & 4561 & 0.848 & 3768 & uncalibrated \\
2 & 4278 & 0.796 & 5955 & uncalibrated \\
3 & 4315 & 0.803 & 3050 & uncalibrated \\
4 & 4236 & 0.788 & 3760 & uncalibrated \\
\bottomrule
\end{tabular}
\end{table}

\section{Additional Experiments}
\label{app:additional_experiments}

Three further controlled experiments not shown in the main-text figures.

\subsection{Coherence and Grouped Reporting (Exp 2)}
\label{subsec:results_h2}
We form increasingly coherent source pairs by shifting a copy of one inventory map by
a controlled spatial offset, so decreasing the shift drives coherence toward one.
Shrinking the offset drives maximum eligible coherence from $0.05$ to $0.96$ while
$\sigma_J$ falls from $21.8$ to $4.6$, yet at every offset the conservative rule kept
the two sources in separate singleton components with exact recovery (individual and
grouped relative error $0$; Table~\ref{tab:results_h2}): high coherence alone did not
force a merge here. The reporting layer retains \emph{every} trigger edge without
deduplication when a merge does fire (an A--B--C chain keeps both edges).
\begin{table}[t]
\centering
\footnotesize
\caption{Experiment 2: coherence rises as the source offset shrinks, but the
conservative rule merged no sources and recovery stayed exact.}
\label{tab:results_h2}
\begin{tabular}{rrrll}
\toprule
offset & max.\ coh. & $\sigma_J$ & merged & indiv./grp.\ err. \\
\midrule
8.0 & 0.048 & 21.77 & no & 0.0 / 0.0 \\
6.0 & 0.862 & 10.48 & no & 0.0 / 0.0 \\
4.0 & 0.905 & 8.40  & no & 0.0 / 0.0 \\
2.0 & 0.945 & 5.43  & no & 0.0 / 0.0 \\
1.0 & 0.958 & 4.64  & no & 0.0 / 0.0 \\
\bottomrule
\end{tabular}
\end{table}

\subsection{Missing-Source Adequacy (Exp 8)}
\label{subsec:results_missing_source}
On the real New Delhi platform with a non-empty rank-4 primary background, the
refitted parametric bootstrap ($\alpha=0.05$, $1000$ replicates) is well-calibrated
(null rejection rate $0.05$) and has full power against a residual-visible omitted
source whose signature lies largely outside
$\operatorname{span}[H_\Phi^{\mathrm{lag}},Q]$ (rejection $1.0$); an \emph{aligned}
omission, a nonnegative combination of the fitted columns and background, is absorbed
and not detected (rejection $0.05$; Table~\ref{tab:results_missing}). The test is
therefore one-sided: a rejection is evidence of model inadequacy, whereas a
non-rejection cannot certify inventory completeness.
\begin{table}[t]
\centering
\footnotesize
\caption{Experiment 8: missing-source adequacy on the platform (non-empty $Q$).
Calibrated null, full power on a residual-visible omission, and no detection of an
absorbed in-span omission.}
\label{tab:results_missing}
\begin{tabular}{lr}
\toprule
omission case & rejection rate \\
\midrule
none (null, calibration) & 0.05 \\
residual-visible (out of span, frac.\ 0.91) & 1.00 \\
aligned (in span, absorbed) & 0.05 \\
\bottomrule
\end{tabular}
\end{table}

\subsection{Per-Sensor Footprints (Exp 10)}
\label{subsec:results_footprints}
On controlled trials the per-sensor footprints of
Section~\ref{sec:identifiability_theory} localize the responsible upwind cells (mass
centroid $0.74$ cells from the true origin, $100\%$ of mass within the $4$-cell
radius, nonnegative) and per-sensor contributions sum to the fitted sensor signal to
machine precision; because a non-singleton component is present, footprints are
reported at group granularity. On observed New Delhi data footprints are reported per
monitor at the identifiable report groups, accompanying the apportionment of
Section~\ref{subsec:results_new_delhi_apportionment}
(Appendix~\ref{app:footprint_details}).

\section{Extended Discussion on Related Work}
\label{app:extended_related_work}

This appendix expands the four strands summarized in
Section~\ref{sec:related_work}.

\noindent\textbf{Air-pollution source apportionment.}
Source apportionment has a long history in air-quality science, with
receptor-oriented methods estimating source contributions from ambient chemical
or physical measurements and source-oriented methods propagating emissions
inventories through atmospheric transport models. Classical receptor frameworks
include chemical mass balance and factor-analytic approaches such as positive
matrix factorization \citep{watson2002receptor,paatero1994positive,hopke2016review}.
Reviews and harmonization efforts have clarified best practices for receptor
models, source-oriented models, and their combined use, while also documenting
practical sensitivity to source profiles, factor interpretation, chemical
resolution, and model comparability
\citep{viana2008source,belis2013critical,belis2019european,mircea2020european,hopke2020global}.
These sensitivities arise because ambient measurements provide only indirect
mixtures of source contributions: source profiles may be incomplete or locally
unavailable, tracers and proxy inventories can be collinear, factor
decompositions may admit multiple plausible interpretations, and transport or
inventory assumptions can map distinct source configurations to similar
observations. This ambiguity is especially visible in settings such as India,
where studies can produce divergent source-contribution estimates for the same
city because of limited local profiles, collinear tracers, and methodological
differences \citep{pant2012critical}; global reviews further show that
policy-relevant source categories such as traffic, industry, dust, and domestic
combustion are routinely reported from heterogeneous apportionment studies
\citep{karagulian2015contributions}. Our work is complementary to this
literature: rather than proposing another receptor model or emissions inventory,
we analyze when known or proxy source maps are distinguishable from sparse
concentration sensors after wind-conditioned transport, finite lags, background
correction, and noise.

\noindent\textbf{Inverse modeling and physics-constrained learning.}
Estimating latent states or parameters from partial observations is central to
inverse problems and data assimilation \citep{tarantola2005inverse,engl1996regularization}.
Classical approaches such as Kalman filtering, ensemble Kalman filters, and
variational data assimilation use dynamical models together with priors,
covariance assumptions, and observation models to infer hidden quantities from
measurements \citep{wang2023kalman,carrassi2018data}. More recent physics-informed
and operator-learning methods, including PINNs and neural operators, incorporate
governing equations or physics-motivated structure into flexible learning systems
for parameter estimation and system identification
\citep{raissi2019physics,li2024physics,li2020fourier,bhardwaj2025fieldformer}.
These methods motivate our use of a transport-structured, PyTorch-based
constrained fitting backend, but they do not by themselves resolve the
identifiability question. In sparse sensing regimes, the map from latent source
activities, transport parameters, and background components to sensor
observations can be many-to-one: different parameterizations may fit the same
measurements after optimization. Our focus is therefore complementary to
calibration and physics-constrained learning methods. Rather than asking only how
to fit the projected sensor data, we characterize when the fitted source--basis
coefficients are identifiable from the projected lagged response matrix and when
the result should be weakened, qualified, or grouped.

\noindent\textbf{Sensor-space spatio-temporal learning.}
Graph neural networks, transformers, and sequence models have been widely used to
model spatio-temporal processes from sensor data, including traffic,
environmental, and air-quality time series
\citep{li2017diffusion,yu2017spatio,wu2019graph,chen2023group,iyer2022modeling,bhardwaj2025comprehensive,bhardwaj2025fieldformer}.
These methods are relevant to our setting because they operate in the same
sparse-sensor observational regime. However, their usual objective is
forecasting, interpolation, or representation learning over sensor measurements
rather than attribution to physically meaningful source inventories. A
sensor-space model can predict observed concentrations well while leaving
unresolved whether distinct source groups would produce distinguishable
wind-conditioned fingerprints at the deployed sensors. Our work is therefore
complementary: instead of learning an unconstrained predictive representation of
the sensor field, we analyze the identifiability of declared source--basis
contributions under a specified transport, background, lag, and noise model.

\noindent\textbf{Observability and system identification.}
The question of whether latent quantities can be recovered from observations is
closely related to classical notions of observability in control theory
\citep{chen1984linear}. For linear dynamical systems, observability characterizes
whether the internal state can be uniquely determined from output trajectories.
Extensions to nonlinear systems and PDEs have been studied in system
identification and inverse-problem literature
\citep{ljung1987theory,banks2012estimation,colton1990inverse}, often under
assumptions of sufficient measurement richness. We bring this perspective to
inventory-based air-pollution source apportionment under sparse spatial sensing.
Rather than asking whether a full latent state or physical model is observable, we
ask whether declared source--basis contributions have distinct projected
sensor-time fingerprints after wind-conditioned transport, finite lags, background
correction, and noise. This leads to diagnostics tailored to the apportionment
task, including rank and singular-value stability, source visibility, pairwise
coherence, background absorption, and conservative grouping, together with a
real-world instantiation that reports the attribution resolution supported by the
available sensors, inventories, and transport assumptions.

\section{Implementation Invariants}
\label{app:implementation_invariants}

The response is constructed first on the complete time-major, sensor-minor row
index.  The PM\(_{2.5}\) observation mask then selects the same ordered rows
from \(Y\), \(H_\Phi^{\mathrm{lag}}\), \(Q\), and metadata.  Source--basis
columns are always source-major and basis-minor.  Any row or column metadata
mismatch is a hard error.

The response construction uses integer source-cell centers, half-cell domain
extents, fractional final advection substeps, a positive minimum dispersion
age, direct sensor evaluation, finite lag, and complete removal after center
exit.  Kernel boundary/truncation loss and whole-release exit loss are recorded
separately.  Retained and dropped kernel mass sum to emitted kernel mass within
the recorded quadrature tolerance; sensor observations are never summed as a
mass proxy.  The default puff baseline has zero source and zero initial state.

Normal backgrounds have effective rank at most eight and do not use inventory
or response columns.  Thin-SVD projection stores the numerical tolerance,
singular values, retained \(U_r\), dependent input columns, removed components,
and orthogonality/idempotence residuals.  A source-like column is permitted only
when the run is labeled as a stress test.

The primary \(Q\), lag-candidate grid, and fixed-zero coefficient set are
declared before fitting.  A fixed-zero mask removes columns before diagnostics
and fitting and stores both reduced-to-original and original-to-reduced index
maps.  Fitted near-zero coefficients never modify the mask.  Adjacent lag
responses use identical rows and columns; the default convergence threshold is
\(\tau_L=10^{-3}\).

\section{Sanity Gates}
\label{app:sanity_gates}

Before paper experiments, the implementation must pass six small deterministic
gates built from public APIs:
\begin{enumerate}[leftmargin=*]
    \item \textbf{S1: response.}  On a \(16\times16\) grid with four compact
    sources, four sensors, impulse/constant bases, and seed 123, verify response
    shape, downwind arrival, open-boundary exit, no wraparound, dispersion
    anisotropy, mass accounting, metadata, and exact repeatability.
    \item \textbf{S2: projection.}  Use the S1 response with a rank-four normal
    background to verify exact ordering, background removal, orthogonality,
    idempotence, redundant-column invariance, and the labeled source-like stress
    behavior.
    \item \textbf{S3: diagnostics.}  Verify orthogonal, duplicate, and weak
    matrices, including padded zeros, \(\sigma_J\), infinite deficient
    condition status, weak-pair N/A values, and a matched multi-source eastward
    versus two-direction wind comparison.
    \item \textbf{S4: fit.}  Recover known nonnegative coefficients in
    noiseless/noisy well-conditioned cases, verify projected-FISTA KKT
    convergence, preserve exact nonnegativity, and recover the duplicate-pair
    sum when the individual split is nonunique.
    \item \textbf{S5: report groups.}  Verify the duplicate-source connected
    component, an \(A\)--\(B\)--\(C\) chain with both trigger edges retained,
    separated-source non-edge, independent weak flag, conservative-group flag,
    and summed group contribution.
    \item \textbf{S6: end to end.}  Build sources, wind, response, background,
    masked observations, diagnostics, fit, uncertainty, and report groups with
    one command and write a compact provenance-complete summary.
\end{enumerate}
Task 10 experiments begin only after S1--S6 and the PyTorch CPU/CUDA parity
checks pass.  Two further gates support the results.  \textbf{S7 (adequacy
calibration)} runs a statistically calibrated study of the residual-adequacy test
on a controlled orthonormal design, checking that the null rejection rate matches
\(\alpha\), that the null \(p\)-values are uniform, and that power increases
monotonically with omission strength.  A final \textbf{reporting gate} exercises
the evaluation layer that emits the result tables in Section~\ref{sec:evaluation},
verifying our reporting rules: undefined weak-pair metrics stay null,
grouped metrics accompany any non-singleton report component, an \(A\)--\(B\)--\(C\)
merge chain retains both trigger edges, percentages are labeled fractions of fitted
sensor signal, and an uncalibrated run never reports an adequacy pass.

\section{PyTorch Solver and Uncertainty Protocol}
\label{app:pytorch_solver}

After CSV/NPZ ingestion, numerical computation uses PyTorch tensors on an
explicit device and dtype.  The nonnegative coefficient objective is
\begin{equation}
f(\mathbf c)=
\|\widetilde Y-\widetilde H_\Phi\mathbf c\|_2^2
+\lambda\|\mathbf c-\mathbf c_0\|_2^2,
\qquad \mathbf c\ge0.
\label{eq:appendix_fista_objective}
\end{equation}
Projected FISTA uses a Lipschitz step derived from \(\sigma_1\) or a recorded
power-iteration estimate, clamps every iterate to the nonnegative orthant, and
restarts acceleration whenever the objective increases.  Termination requires
the configured projected-gradient/KKT residual and relative objective-change
tolerances, or reports that the iteration cap was reached.  The fit artifact
stores convergence status, iteration count, final KKT residual, objective
summary, device, and dtype.

Active-set covariance uses PyTorch linear algebra on
\(\widetilde H_{\Phi,\mathcal A}^{\top}
\widetilde H_{\Phi,\mathcal A}+\lambda I\).  Wind-ensemble and bootstrap refits
are batched on the same device when shapes permit.  Undefined weak-pair metrics
are serialized as JSON \texttt{null}, never NaN or a false numerical distance.

Uncertainty artifacts distinguish \texttt{ensemble\_kind=transport} from
\texttt{ensemble\_kind=inventory}.  Only the former may be summarized as an
ensemble uncertainty interval; inventory members remain named robustness
scenarios, and an aggregation request that mixes the two kinds is an error.

\section{Adequacy and Ensemble Result Contracts}
\label{app:adequacy_contracts}

A residual-adequacy result stores \texttt{calibration\_status},
\(T_{\mathrm{res}}\), bootstrap quantile, Monte Carlo \(p\)-value,
\(\alpha\), replicate count, and noise-model provenance.  Paper runs use
\(\alpha=0.05\) and 1000 refitted replicates.  If the externally calibrated
noise covariance is absent or invalid, numerical test fields are null and the
status is \texttt{uncalibrated}; no boolean pass is emitted.  Sensor, time, and
autocorrelation residual summaries are retained in either case.

A lag-selection result stores every candidate \(L\), adjacent convergence
ratio, selected \(L\), \(\tau_L\), physical grid rationale, row count, and
diagnostic/group summaries.  A wind-distribution result stores the sampled-window
provenance; 5th, 50th, and 95th percentiles of \(\sigma_J\); full numerical and
effective-rank probabilities; coefficient weakness probabilities; pairwise
ambiguity probabilities; and conservative-component frequencies.

\section{Reproducibility and Artifact Provenance}
\label{app:reproducibility}

Every paper run records the container image, repository revision, PyTorch and
CUDA versions, device model, dtype, deterministic-algorithm settings, random
seeds, source files and crop, per-source scales, wind-imputer architecture, held-out wind split, query grid, gridded wind product, and ensemble provenance,
inventory identifier/hash/version, observation mask, response and transport
configuration, wind realization, dispersion settings, background columns and
rank tolerance, lag protocol, fixed-zero mask and index mapping, diagnostic
thresholds, ensemble kind, solver settings, noise-model provenance, and parent
artifact hashes.
Generated wind products, response matrices, checkpoints, and experiment outputs
are written to generated-data or log paths and are not treated as immutable
source files.

Controlled-run artifacts contain true and estimated coefficients, reconstructed
activities, response/projection results, diagnostics, uncertainty, report
groups, and recovery metrics.  Observed New Delhi artifacts omit unavailable
ground truth and instead contain observed-row counts, fitted trajectories, raw
and projected residuals, normalized proxy activities/contributions,
uncertainty, weak/ambiguous flags, trigger pairs, and group summaries.  A
percentage is emitted only with an explicit denominator identifying it as a
fraction of fitted inventory-attributed sensor signal.


\begin{thebibliography}{35}
\providecommand{\natexlab}[1]{#1}

\bibitem[{Banks and Kunisch(2012)}]{banks2012estimation}
Banks, H.~T.; and Kunisch, K. 2012.
\newblock \emph{Estimation techniques for distributed parameter systems}.
\newblock Springer Science \& Business Media.

\bibitem[{Beck and Teboulle(2009)}]{beck2009fast}
Beck, A.; and Teboulle, M. 2009.
\newblock A fast iterative shrinkage-thresholding algorithm for linear inverse
  problems.
\newblock \emph{SIAM Journal on Imaging Sciences}, 2(1): 183--202.

\bibitem[{Belis et~al.(2019)Belis, Favez, Mircea, Diapouli, Manousakas,
  Vratolis, Gilardoni, Paglione, Decesari, Mocnik et~al.}]{belis2019european}
Belis, C.; Favez, O.; Mircea, M.; Diapouli, E.; Manousakas, M.; Vratolis, S.;
  Gilardoni, S.; Paglione, M.; Decesari, S.; Mocnik, G.; et~al. 2019.
\newblock European guide on air pollution source apportionment with receptor
  models—Revised version 2019.
\newblock \emph{Publications Office, LU}.

\bibitem[{Belis et~al.(2013)Belis, Karagulian, Larsen, and
  Hopke}]{belis2013critical}
Belis, C.; Karagulian, F.; Larsen, B.~R.; and Hopke, P. 2013.
\newblock Critical review and meta-analysis of ambient particulate matter
  source apportionment using receptor models in Europe.
\newblock \emph{Atmospheric Environment}, 69: 94--108.

\bibitem[{Bhardwaj et~al.(2025)Bhardwaj, Balashankar, Iyer, Soans, Sudarshan,
  Pande, and Subramanian}]{bhardwaj2025comprehensive}
Bhardwaj, A.; Balashankar, A.; Iyer, S.; Soans, N.; Sudarshan, A.; Pande, R.;
  and Subramanian, L. 2025.
\newblock Comprehensive monitoring of air pollution hotspots using sparse
  sensor networks.
\newblock \emph{ACM Journal on Computing and Sustainable Societies}, 3(4):
  1--36.

\bibitem[{Bhardwaj, Balashankar, and
  Subramanian(2025)}]{bhardwaj2025fieldformer}
Bhardwaj, A.; Balashankar, A.; and Subramanian, L. 2025.
\newblock FieldFormer: Physics-Informed Transformers for Spatio-Temporal Field
  Reconstruction from Sparse Sensors.
\newblock \emph{arXiv preprint arXiv:2510.03589}.

\bibitem[{Carrassi et~al.(2018)Carrassi, Bocquet, Bertino, and
  Evensen}]{carrassi2018data}
Carrassi, A.; Bocquet, M.; Bertino, L.; and Evensen, G. 2018.
\newblock Data assimilation in the geosciences: An overview of methods, issues,
  and perspectives.
\newblock \emph{Wiley Interdisciplinary Reviews: Climate Change}, 9(5): e535.

\bibitem[{{Center for International Earth Science Information Network (CIESIN),
  Columbia University}(2018)}]{ColumbiaUniversity2018}
{Center for International Earth Science Information Network (CIESIN), Columbia
  University}. 2018.
\newblock Gridded Population of the World, Version 4 (GPWv4): Population
  Density, Revision 11.

\bibitem[{Chen(1984)}]{chen1984linear}
Chen, C.-T. 1984.
\newblock \emph{Linear system theory and design}, volume 301.
\newblock Holt, Rinehart and Winston New York.

\bibitem[{Chen et~al.(2023)Chen, Xu, Wu, and Huang}]{chen2023group}
Chen, L.; Xu, J.; Wu, B.; and Huang, J. 2023.
\newblock Group-aware graph neural network for nationwide city air quality
  forecasting.
\newblock \emph{ACM Transactions on Knowledge Discovery from Data}, 18(3):
  1--20.

\bibitem[{Colton et~al.(1990)Colton, Ewing, Rundell et~al.}]{colton1990inverse}
Colton, D.~L.; Ewing, R.~E.; Rundell, W.; et~al. 1990.
\newblock \emph{Inverse problems in partial differential equations}, volume~42.
\newblock Siam.

\bibitem[{CPCB(2023)}]{cpcb}
CPCB. 2023.
\newblock Air Quality Data.
\newblock \url{https://cpcb.nic.in/}.

\bibitem[{CPCB(2025)}]{cpcb_portal}
CPCB. 2025.
\newblock CPCB Data Portal.
\newblock
  \url{https://app.cpcbccr.com/ccr/#/caaqm-dashboard-all/caaqm-landing/caaqm-comparison-data}.

\bibitem[{Engl, Hanke, and Neubauer(1996)}]{engl1996regularization}
Engl, H.~W.; Hanke, M.; and Neubauer, A. 1996.
\newblock \emph{Regularization of inverse problems}, volume 375.
\newblock Springer Science \& Business Media.

\bibitem[{{Google}(2024)}]{google_maps}
{Google}. 2024.
\newblock Google Maps.
\newblock \url{https://maps.google.com}.
\newblock Accessed: 2026-03-30.

\bibitem[{Guttikunda and Calori(2013)}]{GUTTIKUNDA2013101}
Guttikunda, S.~K.; and Calori, G. 2013.
\newblock A GIS based emissions inventory at 1 km × 1 km spatial
  resolution for air pollution analysis in Delhi, India.
\newblock \emph{Atmospheric Environment}, 67: 101--111.

\bibitem[{Hopke(2016)}]{hopke2016review}
Hopke, P.~K. 2016.
\newblock Review of receptor modeling methods for source apportionment.
\newblock \emph{Journal of the Air \& Waste Management Association}, 66(3):
  237--259.

\bibitem[{Hopke et~al.(2020)Hopke, Dai, Li, and Feng}]{hopke2020global}
Hopke, P.~K.; Dai, Q.; Li, L.; and Feng, Y. 2020.
\newblock Global review of recent source apportionments for airborne
  particulate matter.
\newblock \emph{Science of The Total Environment}, 740: 140091.

\bibitem[{Iyer et~al.(2022)Iyer, Balashankar, Aeberhard, Bhattacharyya,
  Rusconi, Jose, Soans, Sudarshan, Pande, and Subramanian}]{iyer2022modeling}
Iyer, S.~R.; Balashankar, A.; Aeberhard, W.~H.; Bhattacharyya, S.; Rusconi, G.;
  Jose, L.; Soans, N.; Sudarshan, A.; Pande, R.; and Subramanian, L. 2022.
\newblock Modeling fine-grained spatio-temporal pollution maps with low-cost
  sensors.
\newblock \emph{npj Climate and Atmospheric Science}, 5(1): 76.

\bibitem[{Karagulian et~al.(2015)Karagulian, Belis, Dora, Pr{\"u}ss-Ust{\"u}n,
  Bonjour, Adair-Rohani, and Amann}]{karagulian2015contributions}
Karagulian, F.; Belis, C.~A.; Dora, C. F.~C.; Pr{\"u}ss-Ust{\"u}n, A.~M.;
  Bonjour, S.; Adair-Rohani, H.; and Amann, M. 2015.
\newblock Contributions to cities' ambient particulate matter (PM): A
  systematic review of local source contributions at global level.
\newblock \emph{Atmospheric environment}, 120: 475--483.

\bibitem[{Li et~al.(2017)Li, Yu, Shahabi, and Liu}]{li2017diffusion}
Li, Y.; Yu, R.; Shahabi, C.; and Liu, Y. 2017.
\newblock Diffusion convolutional recurrent neural network: Data-driven traffic
  forecasting.
\newblock \emph{arXiv preprint arXiv:1707.01926}.

\bibitem[{Li et~al.(2020)Li, Kovachki, Azizzadenesheli, Liu, Bhattacharya,
  Stuart, and Anandkumar}]{li2020fourier}
Li, Z.; Kovachki, N.; Azizzadenesheli, K.; Liu, B.; Bhattacharya, K.; Stuart,
  A.; and Anandkumar, A. 2020.
\newblock Fourier neural operator for parametric partial differential
  equations.
\newblock \emph{arXiv preprint arXiv:2010.08895}.

\bibitem[{Li et~al.(2024)Li, Zheng, Kovachki, Jin, Chen, Liu, Azizzadenesheli,
  and Anandkumar}]{li2024physics}
Li, Z.; Zheng, H.; Kovachki, N.; Jin, D.; Chen, H.; Liu, B.; Azizzadenesheli,
  K.; and Anandkumar, A. 2024.
\newblock Physics-informed neural operator for learning partial differential
  equations.
\newblock \emph{ACM/IMS Journal of Data Science}, 1(3): 1--27.

\bibitem[{Ljung(1987)}]{ljung1987theory}
Ljung, L. 1987.
\newblock \emph{System Identification: Theory for the User}.
\newblock Prentice Hall.

\bibitem[{Mircea et~al.(2020)Mircea, Calori, Pirovano, Belis
  et~al.}]{mircea2020european}
Mircea, M.; Calori, G.; Pirovano, G.; Belis, C.; et~al. 2020.
\newblock European guide on air pollution source apportionment for particulate
  matter with source oriented models and their combined use with receptor
  models.
\newblock \emph{Publications Office of the European Union, LU}.

\bibitem[{Paatero and Tapper(1994)}]{paatero1994positive}
Paatero, P.; and Tapper, U. 1994.
\newblock Positive matrix factorization: A non-negative factor model with
  optimal utilization of error estimates of data values.
\newblock \emph{Environmetrics}, 5(2): 111--126.

\bibitem[{Pant and Harrison(2012)}]{pant2012critical}
Pant, P.; and Harrison, R.~M. 2012.
\newblock Critical review of receptor modelling for particulate matter: a case
  study of India.
\newblock \emph{Atmospheric Environment}, 49: 1--12.

\bibitem[{Raissi, Perdikaris, and Karniadakis(2019)}]{raissi2019physics}
Raissi, M.; Perdikaris, P.; and Karniadakis, G.~E. 2019.
\newblock Physics-informed neural networks: A deep learning framework for
  solving forward and inverse problems involving nonlinear partial differential
  equations.
\newblock \emph{Journal of Computational physics}, 378: 686--707.

\bibitem[{Seber and Lee(2003)}]{seber2003linear}
Seber, G. A.~F.; and Lee, A.~J. 2003.
\newblock \emph{Linear Regression Analysis}.
\newblock Hoboken, NJ: Wiley, 2nd edition.

\bibitem[{Tarantola(2005)}]{tarantola2005inverse}
Tarantola, A. 2005.
\newblock \emph{Inverse problem theory and methods for model parameter
  estimation}.
\newblock SIAM.

\bibitem[{Viana et~al.(2008)Viana, Kuhlbusch, Querol, Alastuey, Harrison,
  Hopke, Winiwarter, Vallius, Szidat, Pr{\'e}v{\^o}t et~al.}]{viana2008source}
Viana, M.; Kuhlbusch, T.~A.; Querol, X.; Alastuey, A.; Harrison, R.~M.; Hopke,
  P.~K.; Winiwarter, W.; Vallius, M.; Szidat, S.; Pr{\'e}v{\^o}t, A.~S.; et~al.
  2008.
\newblock Source apportionment of particulate matter in Europe: a review of
  methods and results.
\newblock \emph{Journal of aerosol science}, 39(10): 827--849.

\bibitem[{Wang et~al.(2023)Wang, Sun, Jiang, Zeng, and Liu}]{wang2023kalman}
Wang, B.; Sun, Z.; Jiang, X.; Zeng, J.; and Liu, R. 2023.
\newblock Kalman filter and its application in data assimilation.
\newblock \emph{Atmosphere}, 14(8): 1319.

\bibitem[{Watson et~al.(2002)Watson, Zhu, Chow, Engelbrecht, Fujita, and
  Wilson}]{watson2002receptor}
Watson, J.~G.; Zhu, T.; Chow, J.~C.; Engelbrecht, J.; Fujita, E.~M.; and
  Wilson, W.~E. 2002.
\newblock Receptor modeling application framework for particle source
  apportionment.
\newblock \emph{Chemosphere}, 49(9): 1093--1136.

\bibitem[{Wu et~al.(2019)Wu, Pan, Long, Jiang, and Zhang}]{wu2019graph}
Wu, Z.; Pan, S.; Long, G.; Jiang, J.; and Zhang, C. 2019.
\newblock Graph wavenet for deep spatial-temporal graph modeling.
\newblock \emph{arXiv preprint arXiv:1906.00121}.

\bibitem[{Yu, Yin, and Zhu(2017)}]{yu2017spatio}
Yu, B.; Yin, H.; and Zhu, Z. 2017.
\newblock Spatio-temporal graph convolutional networks: A deep learning
  framework for traffic forecasting.
\newblock \emph{arXiv preprint arXiv:1709.04875}.

\end{thebibliography}
\end{document}